\documentclass{tlp}

\usepackage{amsmath,amssymb}
\usepackage{url}

\usepackage[bookmarksopen]{hyperref}
\newtheorem{theorem}{Theorem}

\newtheorem{corollary}{Corollary}
\newtheorem{proposition}{Proposition}
\newtheorem{lemma}{Lemma}

\newtheorem{definition}{Definition}

\title{Computing Stable Models of Normal Logic Programs Without Grounding}
\author[Kyle Marple, Elmer Salazar and Gopal Gupta]
{KYLE MARPLE, ELMER SALAZAR and GOPAL GUPTA\thanks{This material is based upon 
		work supported by the National Science Foundation under Grant No. 
		1423419.}\\
University of Texas at Dallas, 800 W. Campbell Rd, Richardson, TX 75080}

\begin{document}

\maketitle

\begin{abstract}
	We present a method for computing stable models of normal logic programs, 
	i.e., logic programs extended with negation, in the presence of predicates 
	with arbitrary terms. Such programs need not have a finite grounding, so 
	traditional methods do not apply. Our method relies on the use of a 
	non-Herbrand universe, as well as coinduction, constructive negation and a 
	number of other novel techniques. Using our method, a normal logic program 
	with predicates can be executed directly under the stable model semantics 
	without requiring it to be grounded either before or during execution and 
	without requiring that its variables range over a finite domain. 
	As a result, our method is quite general and supports the use of terms 
	as arguments, including lists and complex data structures.  A prototype 
	implementation and non-trivial applications have been developed to demonstrate the feasibility of our method.
\end{abstract}

\begin{keywords}
	stable models, predicate, goal-directed
\end{keywords}

\section{Introduction}

The addition of negation to logic programming has been the subject of 
significant research over the last several decades. Both classical negation 
(where negative inference must be explicitly established) and \textit{negation 
as failure} (where \texttt{not p} is inferred if we fail to establish 
\texttt{p}) have been considered. Of these, negation as failure (NAF) has more 
interesting applications and has been widely researched.

Including NAF poses the following problem: while a logic program with no 
negation has a unique minimal model, logic programs with NAF can lead to 
multiple, incompatible models. Different semantics of negation arise depending 
on which of these models are deemed acceptable. These semantics include 
Fitting's semantics, the well-founded semantics of Ross, Gelder and 
Schlipf, the stable model semantics and a multitude of others 
\cite{fitting,wfs,asp,ullman-survey}. Of the various semantics which have been 
developed, the stable model semantics is widely regarded as the most expressive 
\cite{baral}.

However, given the current computation methods, the stable model semantics, and 
the answer set programming paradigm inspired by them, are computable only for 
programs which are finitely groundable. Thus, 
current methods compute the stable models of a program by first 
\textit{grounding} the program, i.e., instantiating each program variable with 
each of the values from its respective domain to derive ground clauses. The 
stable models are then computed using the grounded program. In most modern 
implementations, the ground program is suitably transformed and fed to a SAT 
solver. The models produced by the SAT solver will be the stable models of the 
original program \cite{assat,clasp}.

There are several problems with the grounding-based approach, the most 
significant being that only certain classes of programs are guaranteed to
have a finite grounding. A logic program with even a single unary term will 
have an infinite number of groundings due to the fact that, given a single 
unary function symbol f/1 and a single constant value \texttt{a}, the domain 
over which variables can range is infinite, consisting of \{\texttt{a, f(a), 
f(f(a)), f(f(f(a))), \dots }\}, resulting in infinite number of grounded 
clauses. For such a program to have a finite grounding, each variable must be 
restricted to a finite domain. However, even with finite domains, the grounding 
of a program may be exponentially larger than the original. Secondly, SAT-based 
and similar approaches compute the complete model of the 
grounded program. In reality, when solving practical problems, we are often 
only interested in part of the model. Finally, adding negation can lead to 
(parts of) the program becoming inconsistent. When this occurs, bottom-up 
methods that work with grounded programs will declare the whole program to be 
inconsistent (i.e., no model exists). In practice, it may 
be desirable to compute answers as long as these answers do not involve the 
inconsistent part of the program \cite{dcc}.

The last two problems described above can be resolved by designing 
goal-directed or query-driven execution methods for computing stable models. 
Such a method does not use a SAT solver, but rather, given a query, computes a 
partial stable model containing the query, if one exists. The computation is 
done in a manner very similar to SLD resolution \cite{lloyd87} for logic 
programs \cite{goalasp,dcc}. We have presented such a method previously, 
however, that method only works for propositional (grounded) 
programs\cite{goalasp,dcc}. In this paper, we build on our previous work to 
remove the need for grounding. Thus, we develop a query-driven method that can 
apply the stable model semantics to a normal logic program containing arbitrary 
terms as well as negation. This is accomplished without grounding the program, 
either before or during execution. It should be noted that, until recently, 
such query-driven procedures were considered impossible to develop even for 
propositional logic programs \cite{survey}.

The key insight in our work is that we use a non-Herbrand universe (in fact, an 
infinite \textit{superset} of the Herbrand universe) which allows us to 
guarantee properties required for the correctness of our method while still 
obtaining useful results. Additionally, coinductive logic programming is used 
to establish the consistency of mutually dependent (co-)recursive calls 
\cite{iclp2007}. \textit{Dual rules} are used both to simplify the handling of 
negation and to provide constructive negation.

It should be noted that top-down, goal-directed implementations which support 
predicates have been designed for the well-founded semantics by extending 
Prolog systems, for example, with \textit{tabling} \cite{xsb-wfs}. However, the 
well-founded semantics can be too weak for many applications, as it declares 
the truth value of many interesting atoms to be unknown. The stable model 
semantics is more expressive, but, to date, there has been no satisfactory 
solution to the problem of computing stable models of arbitrary predicate logic 
programs. Those solutions that have been proposed either greatly restrict the 
types of programs that can be handled or ground the program incrementally 
during execution \cite{dalpalu09,daotran12,lefevre09a,lefevre09b}.
Thus, our research makes important contributions:
\begin{itemize}
	\item It presents a top-down, query-driven method that can execute normal 
	logic programs with arbitrary predicates, thus solving a problem that was 
	hitherto considered unsolvable.
	\item Our method can be thought of as providing an operational semantics to
	normal logic programs with predicates (or, Prolog with negation) under the 
	stable model semantics. This can be combined with other advanced features 
	of 	logic programming such as constraints \cite{clpr} to develop 
	extremely powerful 	applications in an elegant manner, such as automated 
	planning under real-time constraints \cite{Bansal10}.
	\item The stable model semantics and answer set programming have been shown 
	to support powerful reasoning techniques such as default reasoning, 
	counterfactual reasoning, abductive reasoning, etc. These reasoning 
	capabilities now become available within Prolog. 
\end{itemize}

The only restrictions our method places upon programs are that operands of 
arithmetic operations must be ground, two \textit{negatively constrained 
variables} (discussed in Section \ref{sec:variables}) cannot be 
\textit{disunified} with each other and left recursion cannot lead to success. 
Of these, the last restriction can be removed via tabling \cite{swift}.
We prove the soundness of our method for all legal programs and a prototype 
implementation is available \cite{saspweb}. For convenience, we will refer to 
our method by the name given to its prototype implementation: \textbf{s(ASP)}.

We will begin with an overview of our method for the goal-directed 
execution of propositional logic programs and then examine the changes needed 
to adapt it to predicate logic programs containing arbitrary terms. Therefore, 
the remainder of the paper is structured as follows. Section 
\ref{sec:background} contains an overview of the stable model semantics as well 
as our propositional, query-driven execution method. Section \ref{sec:sasp} 
covers the expansion to predicate logic programs. Section \ref{sec:correct} 
contains proofs for the correctness for our method. Section \ref{sec:imp} 
provides a brief discussion of our prototype implementation along with examples 
of its execution. Section \ref{sec:related} discusses related and future work. 
Finally, in Section \ref{sec:conclusion}, we review the paper and draw 
conclusions.

\section{Background} \label{sec:background}

In this section, we will provide background information needed to understand 
both our method and its significance. We begin with an overview of the 
stable model semantics before examining our method for propositional programs, 
upon which our method for predicate programs is built.

\subsection{The Stable Model Semantics} \label{sec:stable}

The stable model semantics provides an intuitive way to represent non-monotonic 
or common sense reasoning in normal logic programs. The building blocks of such 
programs are literals.
\begin{definition} \label{def:literal}
	A \textbf{positive literal} is an atom or compound term, optionally 
	prefixed with a hyphen (`-'). A negated or \textbf{negative literal} is a 
	positive literal preceded by \texttt{not}. The basic term \textbf{literal} 
	encompasses both positive and negative literals.
\end{definition}
As we are dealing with predicate programs, it should be noted that, in the 
case of compound terms, literals with the same functor, arity (number 
of arguments) and negation are treated as instances of the same literal. For 
instance, \texttt{p(X)} and \texttt{p(1)} are both instances of the literal 
\texttt{p/1}, while \texttt{not p(1)} is an instance of \texttt{not p/1}.

\begin{definition} \label{def:normal}
	A \textbf{normal logic program} is a program consisting of clauses of the 
	following forms:\\
	\texttt{
		\indent p :- q$_1$, ..., q$i$, ..., q$_m$,\\
		\indent~~~~~not r$_1$, ..., not r$_j$, ..., not r$_n$.\\
		\indent p :- q$_1$, ..., q$i$, ..., q$_m$,\\
		\indent~~~~~not r$_1$, ..., not r$_j$, ..., not r$_n$,\\
		\indent~~~~~not p.\\
		\indent :- q$_1$, ..., q$i$, ..., q$_m$,\\
		\indent~~~not r$_1$, ..., not r$_j$, ..., not r$_n$.\\
	}
	\noindent where \texttt{m} $\geq$ 0, \texttt{n} $\geq$ 0, $0 \leq i \leq m$ 
	and $0 \leq j \leq n$. Each \texttt{p, q$_i$} and \texttt{not r$_j$} is a 
	literal.
	The literal to the left of the consequence operator (\texttt{:-}) is the 
	\textbf{head} of the clause, while the literals to the right, also referred 
	to 	as \textbf{goals}, form the \textbf{body}. A \textbf{rule} is the set 
	of all clauses in a program with the same literal as their head. Either the 
	head or body of a rule may be empty, resulting in \textbf{headless rules} 
	and \textbf{facts}, respectively.
\end{definition}
\noindent Unless otherwise stated, all programs referenced in this paper will 
be normal logic programs. Two forms of negation are permitted: \textit{negation as failure} and \textit{classical negation}.
\begin{definition} \label{def:naf}
	Under \textbf{negation as failure (NAF)}, \texttt{not p} succeeds iff 
	\texttt{p} fails.
\end{definition}
\noindent Negative literals in the body of a rule are negated using NAF. The 
optional hyphen prefix indicates classical negation.
\begin{definition} \label{def:classical}
	Under \textbf{classical negation}, \texttt{-p} and \texttt{p} cannot both be true, but it is possible for both to be false. Additionally, while \texttt{not p} may only appear in the body of a clause, \texttt{-p} may be used as the head of a clause, e.g.\\
	\texttt{\indent -p :- q, not r.}\\
	Thus, the difference between classical negation and NAF is that classical negation can be used to define explicit rules for establishing falsehood.
\end{definition}
\noindent Both NAF and classical negation carry the implicit constraint that a 
call and its negation cannot both succeed. That is, the success of \texttt{p} 
excludes \texttt{not p} and \texttt{-p}, and the success of either \texttt{not 
p} or \texttt{-p} will exclude \texttt{p}.

From an external standpoint, only positive literals may appear in the head of a 
rule. Internally generated dual rules, discussed in Section \ref{sec:conneg}, 
have negated literals as their heads, but such rules can never be supplied in 
an input program. Note, however, that classical negation does not affect 
whether a literal is positive or negative. Thus \texttt{-p} is considered a positive literal, with \texttt{not -p} its negation.

The stable model semantics of a normal logic program \texttt{P} is defined 
in terms of the stable models of the program \texttt{ground(P)}, obtained by 
grounding the variables in \texttt{P} (note that the grounding is not required 
to be finite). The stable models of \texttt{ground(P)} are traditionally 
identified using the Gelfond-Lifschitz method \cite{baral}. At the most basic 
level, the Gelfond-Lifschitz method is similar to our own: a candidate model is 
generated and then tested to ensure that it is a stable model.

\begin{definition} \label{def:canmodel}
	A \textbf{candidate stable model} is a set of literals which is assumed to 
	be a stable model of a given program.
\end{definition}

\begin{definition}
	\textbf{Gelfond-Lifschitz method (GL Method)}: Given a grounded 
	program \texttt{P} and a candidate stable model \texttt{M}, a residual 
	program \texttt{R} is obtained by applying the following transformation 
	rules: 
	\begin{enumerate}
		\item For all literals \texttt{L $\in$ M}, delete all clauses in 
		\texttt{P} which have \texttt{not L} in their body.
		\item Delete all remaining negated goals (of the form \texttt{not 
		X}) from the bodies of the remaining clauses.
	\end{enumerate}
	Next, the least fixed-point, \texttt{F}, of the residual program 
	\texttt{R} is computed. If \texttt{F $=$ M}, then \texttt{M} is a 
	\textbf{stable model} of \texttt{P}.
\end{definition}

The non-monotonicity of the stable model semantics is captured by the last two 
rule forms given in Definition \ref{def:normal}. Consider the rule:
\begin{verbatim}
p :- q, not p.
\end{verbatim}
Following the Gelfond-Lifschitz method as outlined above, this clause 
restricts \texttt{q} (and \texttt{p}) to not be in the stable model (unless 
\texttt{p} happens to be in the model via another clause, in which case, due to 
presence of \texttt{not p}, this clause will be removed while generating the 
residual program). Note that even though a program can have other rules to 
establish that \texttt{q} is in the stable model, adding the rule above forces 
\texttt{q} to not be in the model unless \texttt{p} succeeds through another 
clause, thus making the stable model semantics non-monotonic.

\begin{figure}
	\figrule
	\begin{center}
		\begin{verbatim}
		% Given 3 birds, which can fly?
		penguin(sam).       % sam is a penguin
		wounded_bird(john). % john is wounded
		bird(tweety).       % tweety is just a bird
		
		% penguins and wounded birds are still birds
		bird(X) :- penguin(X).
		bird(X) :- wounded_bird(X).
		
		% penguins and wounded birds are abnormal
		ab(X) :- penguin(X).
		ab(X) :- wounded_bird(X).
		
		% birds can fly if they are not abnormal
		flies(X) :- bird(X), not ab(X).
		
		% explicit closed world assumptions
		-flies(X) :- ab(X).
		-flies(X) :- -bird(X).
		
		-wounded_bird(X) :- not wounded_bird(X).
		-penguin(X) :- not penguin(X).
		-ab(X) :- not ab(X).
		-bird(X) :- not bird(X).
		\end{verbatim}
	\end{center}
	\caption{A version of the classic ``Tweety Bird'' problem with a 
		combination of classical negation and NAF.}
	\label{fig:tweety}
	\figrule
\end{figure}

Sample normal logic programs are shown in Figures \ref{fig:tweety}, 
\ref{fig:queens} and \ref{fig:hamiltonian}.
The program in Figure \ref{fig:tweety} shows an example of default reasoning 
with the \textit{closed world assumption (CWA)}. Under the CWA, if we do not 
know a piece of information, we infer it to be false. With this program, if we 
pose the query \texttt{?- -flies(X)}, we should get the answers \texttt{X = 
sam} and \texttt{X = john}.

\subsection{Our Method for Propositional Programs} \label{sec:grounded}

Now that we have introduced the basics of the stable model semantics, we can 
discuss our method for goal-directed execution of propositional 
programs, which can be viewed as a stepping stone to our method for 
predicate programs. Our method for propositional programs has been proven sound 
and complete with respect to the 
Gelfond-Lifschitz method and forms the core of the Galliwasp ASP 
system (\url{http://galliwasp.sourceforge.net})
\cite{galliwasp,galliwaspweb,kylethesis}. The two key aspects of our 
propositional method are its handling of rules containing odd loops over
negation and its use of {\it coinduction}.

Both our propositional and predicate methods categorize rules by examining the 
call graph and checking the number of negations between any recursive calls.
\begin{definition}
	A program's \textbf{call graph} is a directed, weighted graph with one node 
	for each positive literal in the program. Edges are drawn from rule heads 
	to their goals. While only positive literals are used as nodes, negation is 
	preserved using weighted edges: edges corresponding to a positive literal 
	are given a weight of 0, while edges corresponding to negative literals are 
	given a weight of 1. To keep track of which rules are part of a given 
	cycle, each edge is also paired with an ID indicating which rule produced 
	it.
\end{definition}
\noindent First, the call graph is traversed to identify any odd loops over 
negation.
\begin{definition}
	An \textbf {odd loop over negation (OLON)} occurs when a cycle in the call 
	graph contains an odd number of negations.
\end{definition}
\noindent Each rule in the program is then classified using the following 
definitions:
\begin{definition}
	An \textbf{OLON rule} is a rule which can be called as part of an OLON.
\end{definition}
\begin{definition}
	\textbf{Ordinary rules} have at least one path in the call graph which will 
	not result in an odd loop over negation.
\end{definition}
\noindent Note that rules with an empty head are always treated as OLON rules.  
Additionally, a rule can be both an OLON rule and an ordinary rule via 
different paths in the call graph.

OLON rules are important to the stable model semantics because they have the 
ability to place global constraints on a program. These constraints must be 
satisfied by any stable model, even if the OLON is never reached during 
execution. Consider the following two forms of OLON rules:
\begin{verbatim}
p :- B, not p.
:- B.
\end{verbatim}
where \texttt{B} is some conjunction of goals. For the first rule, any stable 
model must satisfy one of two cases: (i) \texttt{p} is added to the model by 
another rule in the program, or (ii) at least one goal in \texttt{B} must fail. 
That is, the rule imposes the global constraint \texttt{p} $\vee$ \texttt{not 
B}. For headless rules (a shorthand for the second form), the second case must 
always hold, imposing the global constraint \texttt{not B}.

Programs are executed using a modified 
form of coinduction extended with negation. First, a query is extended to 
enforce the constraints imposed by any OLON rules in the program. Then, this 
query is executed using a modified form of coinductive SLD resolution 
\cite{iclp2007}.

\begin{definition}
	Under \textbf{SLD resolution} \cite{lloyd87}, query is executed by calling 
	each goal in 
	turn. Calls are added to the call stack and expanded by selecting clauses 
	whose head unifies with the call and recursively calling the goals in the 
	body. A call succeeds when this expansion becomes empty (the call or its 
	children unify with facts). If no expansion is possible, backtracking 
	occurs: execution is rolled back to the previous expansion operation and 
	the call is expanded using the next matching clause. A call fails when no 
	matching clauses remain. Execution succeeds when every goal of the query 
	has succeeded and fails when both expansion and backtracking are impossible.
\end{definition}

\begin{definition}
	\textbf{Coinductive SLD resolution (co-SLD resolution)} expands 
	SLD-resolution by storing each succeeding call in a set called the 
	\textbf{coinductive hypothesis set (CHS)}. If a call unifies with a call 
	that is already in the CHS, or with an ancestor in the call stack, the call 
	is allowed to \textit{coinductively succeed} without further expansion 
	\cite{iclp2007}.
\end{definition}

Under the stable model semantics, the condition for coinductive success via the 
call stack is modified such that only cycles containing \textit{even loops} may 
succeed. This modification is necessary because the stable model semantics 
requires that \textit{positive loops} fail, while traditional coinduction would 
allow them to succeed.

\begin{definition} \label{def:even}
	An \textbf{even loop} occurs when a recursive call is encountered with an 
	even, non-zero number of negations between the call and its ancestor in the 
	call stack.
\end{definition}
\begin{definition} \label{def:positive}
	A \textbf{positive loop} occurs when a recursive call is encountered with 
	no negations between the call and its ancestor in the call stack.
\end{definition}

Our methods also add the idea of \textit{coinductive failure}, in which 
failure and backtracking occur if the negation of a call unifies with a call in 
the call stack or CHS. This ensures that the CHS remains consistent, as 
\texttt{p} and \texttt{not p} can never be present at the same time.

Under our methods, the CHS also serves as a \textit{candidate partial 
model}, or candidate model for simplicity. These are conceptually the same as 
the candidate stable models described in Definition \ref{def:canmodel}, except 
that our method focuses on finding subsets of stable models rather than 
complete models (see Definition \ref{def:partial}).

Candidate partial models are generated by executing the ordinary rules in a 
program and then testing to ensure that they satisfy any constraints imposed by 
OLON rules. This testing is handled by the \textit{non-monotonic reasoning 
	check}.
\begin{definition}
	The \textbf{non-monotonic reasoning check (NMR check)} is a special rule 
	responsible for applying the constraints imposed by OLON rules. A call to 
	the NMR check is automatically appended to each query.
\end{definition}

\noindent For each OLON rule in a program, a ``sub-check'' rule with a unique 
head is created to apply the corresponding global constraint (discussed earlier 
in the section). The head of each sub-check rule is then added to the body of 
the NMR check. Sub-check rules are created by adding the negation of the 
corresponding OLON rule's head to the body (if not already present) and then 
negating the rule. Each clause is processed 
independently, so no modification is needed if a goal appears in multiple OLONs 
or as the head of multiple OLON rules. For instance, the rules
\begin{verbatim}
p :- B, not p.
p :- not q, not p.
p :- q, r, not p.
\end{verbatim}
would produce the sub-check rules
\begin{verbatim}
chk_p1 :- not B.
chk_p1 :- p.
chk_p2 :- q.
chk_p2 :- p.
chk_p3 :- not q.
chk_p3 :- not r.
chk_p3 :- p.
\end{verbatim}
As a result, if the NMR check succeeds, the candidate partial model in the CHS 
must satisfy every OLON rule in the program. Correspondingly, if a program or 
candidate partial model is inconsistent, the NMR check will trigger failure and 
backtracking. For example, a program containing the rule
\begin{verbatim}
:- not c.
\end{verbatim}
where \texttt{c} does not appear anywhere else in the program will have no 
stable model. The method enforces this by creating the NMR sub-check
\begin{verbatim}
chk :- c.
\end{verbatim}
Since the program contains no rules for \texttt{c}, the check is unsatisfiable 
and execution will eventually fail.

Upon successful execution of both the query and NMR check, the CHS will be 
returned as a partial stable model.
\begin{definition} \label{def:partial}
	A \textbf{partial stable model} is a set of literals which is guaranteed to 
	be a subset of some stable model of the program \cite{goalasp}.
\end{definition}

\section{The s(ASP) Method} \label{sec:sasp}

Now that we have introduced the stable model semantics and our goal-directed 
method for propositional programs, we can discuss our predicate method. We 
will begin with some fundamentals before looking at several key aspects of the 
method itself. Finally, we will give an overview of the completed method.

\subsection{s(ASP) Fundamentals} \label{sec:fundamentals}

Before we can discuss the actual execution of our method, we must first 
introduce a number of core concepts. These include s(ASP)'s universe, its 
system of variables and constraints, its use of constructive negation and its 
restrictions on legal programs.

\subsubsection{The s(ASP) Universe} \label{sec:universe}

One of the defining aspects of s(ASP) is its universe. While most logic 
programming semantics utilize the Herbrand universe, it is insufficient for our 
purposes. While the Herbrand universe may be finite, s(ASP) explicitly requires 
that its universe always be infinite, the reasons for which will be discussed 
in Section \ref{sec:variables}. To ensure that this property always holds, we 
rely on a universe defined as follows:
\begin{definition} \label{def:universe}
	The \textbf{s(ASP) universe}, \textbf{$U_S$}, is an infinite, proper 
	\textit{superset} of the Herbrand universe.
\end{definition}

Formally, this is achieved by extending the language of propositional 
stable models with an infinite number of special constants in the manner of 
Shoenfield \cite[p. 46]{shoenfield} Shoenfield works with first-order 
mathematical logic rather than non-monotonic logic programming, but translation 
between the two is straightforward: Shoenfield's definition of a first order 
theory corresponds to a system consisting of a language, its universe, a 
semantics, and a program to be executed. As Shoenfield proves, a first order 
theory whose language is extended using his special constants technique is a 
``conservative extension'' of the original theory. That is, a formula (rule) 
from the original theory will hold in the extended theory iff it holds in the 
original theory. In simpler terms, \textit{our extension of the universe does 
not affect the correctness of programs which have been grounded over the 
Herbrand Universe.} Thus, for simplicity, subsequent references to the stable 
model semantics of propositional programs and the GL method will refer to 
variants which have been extended to use the s(ASP) universe.

\subsubsection{Variables and Constraints} \label{sec:variables}

The most obvious step in supporting predicate programs is to accommodate 
variables. Additionally, the constructive negation employed by s(ASP) relies on 
extending variables with simple constraints. In turn, unification and 
\textit{disunification} must be modified to accept such variables.

Instead of the traditional states of bound and unbound, s(ASP) variables can be 
either bound or negatively constrained.
\begin{definition} \label{def:convar}
	A \textbf{negatively constrained variable} is associated with a 
	\textit{prohibited value list}---a list of prohibited values---and 
	represents the set of all values in the s(ASP) universe which are not in 
	this list.
\end{definition}

Thus, if the prohibited value list of a variable {\tt X} contains the
constants {\tt a} and {\tt b}, then {\tt X} may be bound to any value 
\textit{except} {\tt a} and {\tt b}. Unbound variables are treated as a special 
case of negatively constrained variables in which the prohibited value list is 
empty. Note that values in a prohibited value list need not be ground, but must 
be at least partially bound. The presence of a negatively constrained variable
in a prohibited value list would violate our restriction against disunifying 
two negatively constrained variables, discussed further in Section 
\ref{sec:restrictions}.

As discussed in Section \ref{sec:universe}, the s(ASP) universe is an infinite 
superset of the Herbrand universe. The discrepancy between variables with
finite domains and variables defined in terms of the s(ASP) universe can lead 
to what we will refer to as empty variables.
\begin{definition} \label{def:empty}
	A variable which is empty with respect to some domain \texttt{D}, referred 
	to as an \textbf{empty variable} for simplicity, is a negatively 
	constrained variable whose prohibited value list contains every element of 
	\texttt{D}.
\end{definition}

This brings us to the reason that s(ASP) requires a special universe. Clearly, 
were a program to be grounded over some \texttt{D}, no value could be assigned 
to such a variable. However, \textit{it is impossible to ground a program over 
the s(ASP) universe}. Because it is an infinite, proper superset of the 
Herbrand universe, the s(ASP) universe will always contain elements which are 
not present in a given grounding. This provides us with several important 
properties:
\begin{proposition}[Properties of s(ASP) Variables] \label{th:varproperties}
	\begin{enumerate}
		\item It is impossible for a legal program to negatively constrain a 
		variable against every element of the s(ASP) universe.
		\item A variable can never be empty with respect to the s(ASP) universe 
		itself.
		\item The domain of a negatively constrained variable defined in terms 
		of the s(ASP) universe will always be \textit{infinite}.
	\end{enumerate}
\end{proposition}

\begin{proof}
	All of these properties may be derived directly from Definition 
	\ref{def:universe}, Definition \ref{def:empty} and our restriction against 
	disunifying two negatively constrained variables (Section 
	\ref{sec:restrictions}).
\end{proof}

As a result of these properties, variables which are empty with respect to the 
Herbrand universe, or some subset of it, do not trigger failure. For example, 
consider the following program, with rules for negation (termed dual rules, 
discussed in Section \ref{sec:conneg}) added for clarity:
\begin{verbatim}
d(1).
p(X) :- not d(X).
not d(X) :- X \= 1.
not p(X) :- d(X).
\end{verbatim}
If we assume \texttt{D} to be \{1\}, then a grounding of this program would be
\begin{verbatim}
d(1).
p(1) :- not d(1).
not d(1) :- 1 \= 1.
not p(1) :- d(1).
\end{verbatim}
Given the query \texttt{?- p(X).} the grounded program will always fail. 
However, execution of the predicate program with s(ASP) will succeed, 
returning the solution \{ \texttt{p(X)}, \texttt{not d(X)} (\texttt{X} 
\textbackslash= \texttt{1}) \}.

\textit{While results involving empty variables may be different from those of 
a corresponding grounded program, the two will never be inconsistent with 
each other.} Instead, they simply convey different, equally correct 
information. The failure of the grounded program indicates that no solution 
exists for the domain \{1\}, while the success of the predicate program 
indicates that a solution would exist if the domain were to be extended.

\subsubsection{Constructive Unification and Disunification} 
\label{sec:unification}

Now that we have introduced negatively constrained variables, unification and 
disunification must be extended to work with them. To differentiate the 
modified versions from the originals, we will refer to them as 
\textit{constructive} unification and disunification. For cases where neither 
argument contains a negatively constrained variable, the constructive 
algorithms are identical to the traditional ones.

\begin{definition}
	The cases for \textbf{constructive unification} are as follows:
	\begin{itemize}
		\item Constructive unification of a negatively constrained variable with
		a non-variable value will succeed if the non-variable value does not 
		constructively unify with any element in the variable's prohibited 
		value list.
		\item Constructive unification of two negatively constrained 		
		variables will always succeed, setting their shared prohibited value 
		list to the union of their original lists.
		\item Constructive unification of two compound terms is performed 
		recursively: first, the functors and arities are tested, then each pair 
		of corresponding arguments is constructively unified.
		\item In cases where neither argument contains a negatively constrained 
		variable, the result is identical to that of traditional unification.
	\end{itemize}
\end{definition}

Thus, if the prohibited value list of {\tt X} contains {\tt a}, and the 
prohibited value list of {\tt Y} contains {\tt b}, then unifying them, {\tt 
X = Y} where {\tt =} is the unification operator, will extend the prohibited 
value list of both variables to {\tt [a, b]}. After the operation, any 
subsequent attempt to unify either variable with \texttt{a} or \texttt{b} will 
fail.

\begin{definition}
	\textbf{Constructive disunification} is the dual of constructive 
	unification with one exception: in accordance with the restrictions given 
	in Section \ref{sec:restrictions}, constructive disunification of two 
	negatively constrained variables will produce an error. The remaining cases 
	are as follows:
	\begin{itemize}
		\item Constructive disunification of a negatively constrained variable 
		and a non-variable value will always succeed, adding the non-variable 
		value to the variable's prohibited value list.
		\item Constructive disunification of two compound terms is performed 
		by first testing functors and arities. If either of these does not 
		match, the operation succeeds deterministically. Otherwise, the pairs 
		of corresponding arguments are disunified recursively. 
		\textit{Non-deterministic} success occurs as soon as the operation 
		succeeds for a pair of arguments, with subsequent pairs tested upon 
		backtracking.
		\item In cases where neither argument contains a negatively constrained 
		variable, the result is identical to that of traditional disunification.
	\end{itemize}
\end{definition}

Given a variable {\tt X} whose prohibited value list contains \texttt{a}, 
disunifying \texttt{X} with a constant {\tt c}, i.e., solving {\tt X} 
\textbackslash= {\tt c} where \textbackslash\texttt{=} represents the 
disunification operator, will extend the prohibited value list of {\tt X} to 
{\tt [a,c]}, i.e., {\tt X} cannot be bound to either {\tt a} or {\tt c}. Under 
our program restrictions, discussed further in Section \ref{sec:restrictions}, 
the disunification of two negatively constrained variables is considered 
illegal. There is, however, an exception to this rule involving variables in 
even loops, discussed in Section \ref{sec:even}.

Note that constructive disunification of compound terms has the potential to be 
non-deterministic. Consider the following statement:
\begin{verbatim}
a(X, Y) \= a(1, 2)
\end{verbatim}
This operation can succeed for both \texttt{X} \textbackslash= 1 and \texttt{Y} 
\textbackslash= 2, but applying both constraints at the same time would result 
in incompleteness, excluding cases such as \texttt{a(1, 3)}. To preserve 
correctness, the operation will succeed non-deterministically, first succeeding 
for \texttt{X} \textbackslash= 1 and then for \texttt{Y} \textbackslash= 2 upon 
backtracking.

By construction, constructive unification and disunification are sound and 
complete with respect to their traditional counterparts for all legal cases. 
However, as constructive disunification of compound terms may be 
non-deterministic, backtracking may be required to produce all values.

\subsubsection{Constructive Negation and Dual Rules} \label{sec:conneg}

One of the cornerstones of s(ASP) is constructive negation. The stable model 
semantics relies on negation as failure (NAF), which normally returns no 
bindings or other information from a failed call. However, a goal-directed 
implementation of the stable model semantics needs to know \textit{why} a call 
failed rather than simply that it failed. With constructive negation, negated 
calls can be treated the same way as non-negated calls, binding variables and 
contributing to the model. To implement constructive negation, s(ASP) computes 
the \textit{completion} of the program, extending it with rules for negative 
information \cite{lloyd87}. These new rules are called \textit{dual rules}, as 
they represent the negations, or duals, of the rules in the original program 
\cite{alferes}.
\begin{definition} \label{def:dual}
	\textbf{Dual rules} are rules for the negation of a predicate which will 
	succeed whenever a call to the original predicate would fail under NAF.
\end{definition}
In cases without variables or side effects, a predicate's dual rule can be 
computed by simply applying DeMorgan's laws:
\begin{align}
\neg(P \wedge Q) \Leftrightarrow (\neg P) \vee (\neg Q) \\
\neg(P \vee Q) \Leftrightarrow (\neg P) \wedge (\neg Q)
\end{align}
For example, given a predicate \texttt{p} with the clauses
\begin{verbatim}
p :- a, not b.
p :- r.
\end{verbatim}
\noindent the dual of \texttt{p} would be
\begin{verbatim}
not p :- np1, np2.
np1 :- not a.
np1 :- b.
np2 :- not r.
\end{verbatim}

Although the above method works for propositional programs, it requires modification 
to work with predicate programs. To begin with, we will examine two cases that 
must be addressed to account for the introduction of unification. First, 
unification is performed between each of a call's arguments and the 
corresponding arguments in the head of the selected clause. For example, 
consider a call \texttt{t(X, 2)} and a clause with the head \texttt{t(A, B)}. 
When expanding the call using this clause, \texttt{X} will be unified with 
\texttt{A} and \texttt{B} will be unified with 2. Second, the 
presence of variables prevents the goals in a dual rule from being considered 
independently. Consider the clause
\begin{verbatim}
p(X) :- q(X), r(X).
\end{verbatim}
\noindent Our above method would produce the dual
\begin{verbatim}
not p(X) :- not q(X).
not p(X) :- not r(X).
\end{verbatim}
\noindent However, \texttt{q(X)} can modify the value of \texttt{X}, affecting 
the outcome of \texttt{r(X)} and, by extension, \texttt{not r(X)}. Correct dual 
rules must account for this. Finally, we will look at how to handle the 
implicit quantifiers on each variable in a clause.

While the unification that occurs between a call and a rule's head when 
expanding a call is easy to overlook, dual rules must account for it. Since a 
dual rule must succeed when the original would fail, it must also succeed when 
these unification operations would fail. We address this by abstracting such 
operations to remove them from the head and make them explicit. Before 
computing the dual of a clause, the arguments in the head are examined. If a 
variable occurs more than once in the head of a clause, each occurrence after 
the first is replaced with a new variable and a goal unifying the original 
variable and the new one is added to the beginning of the rule body. If an 
argument is a non-variable, it is replaced with a new variable and a goal 
unifying the new variable to the non-variable value is added to the beginning 
of the rule body. So, the rule
\begin{verbatim}
t(A, A).
\end{verbatim}
would first be transformed into
\begin{verbatim}
t(A, B) :- A = B.
\end{verbatim}
which would produce the dual rule
\begin{verbatim}
not t(A, B) :- A \= B.
\end{verbatim}
The operations are the same for compound terms, except that they are performed 
recursively on the arguments of the terms.

Unification also prevents us from considering each goal in a clause 
independently when constructing duals. Because each goal has the potential to 
alter any non-ground variable that it is called with, each goal in a clause may 
depend on the goals called before it. Therefore, before calling the negated 
goal itself, each dual clause must call any goals on which the current goal 
depends. While it should be possible to determine the dependencies of a given 
goal, our current strategy is to simply include in each clause every goal prior 
to the current one. This means that the rule
\begin{verbatim}
p(X, Y) :- not q(X), t(Y, Y).
\end{verbatim}
would produce the dual rule
\begin{verbatim}
not p(X, Y) :- q(X).
not p(X, Y) :- not q(X), not t(Y, Y).
\end{verbatim}

Our algorithm for computing dual rules requires one more modification to 
properly handle predicate programs. This is due to the fact that variables 
have implicit quantifiers: variables in the head of a clause are universally 
quantified, while variables in the body (body variables) are existentially 
quantified. Thus, the clause
\begin{verbatim}
q(X) :- not p(X, Y).
\end{verbatim}
is equivalent to the formula
\begin{equation*}
\forall X(q(X) \leftarrow \exists Y \neg p(X, Y))
\end{equation*}
\noindent Duals of clauses that contain variables must negate these quantifiers 
as well. The universal and existential quantifiers are duals of each other, so 
the dual of the above formula is
\begin{align*}
&\forall X(\neg(q(X)) \leftarrow \neg(\exists Y \neg p(X, Y))) \\
\equiv &\forall X(\neg q(X) \leftarrow \forall Y \neg(\neg p(X, Y))) \\
\equiv &\forall X(\neg q(X) \leftarrow \forall Y p(X, Y))
\end{align*}

\begin{definition}
	\textbf{Body variables} are variables which occur in the body of a clause, 
	but not in its head.
\end{definition}
\noindent As this example shows, any body variables in a clause will be 
universally quantified in the clause's dual. That is, for some binding of the 
head variables, the clause must succeed for all values (or combinations of 
values) of the body variables. To create such duals, we developed a special 
for-all mechanism which relies on negatively constrained variables.
\begin{definition}
	A \textbf{for-all} is a s(ASP) goal of the form \texttt{forall(V, G)} where 
	\texttt{V} is a variable and \texttt{G} is a goal. A for-all succeeds if 
	\texttt{G} succeeds for all values of \texttt{V}.
\end{definition}
These for-alls are used to create dual rules that universally quantify body 
variables from the original rule. First the dual rules for a clause are 
computed as if no body variables were present, except that the predicate in the 
head is replaced with a new, unique one. Next, the body variables are added to 
the head of each dual. Finally, a clause for the dual is created containing 
a for-all over the new predicate. So, the rule
\begin{verbatim}
q(X) :- not p(X, Y).
\end{verbatim}
would produce the dual rule
\begin{verbatim}
not q(X) :- forall(Y, nq1(X, Y)).
nq1(X, Y) :- p(X, Y).
\end{verbatim}
Note that multiple body variables can be handled by nesting for-alls. For 
example \texttt{forall(X, forall(Y, p(X, Y)))} will succeed if \texttt{p(X, Y)} 
succeeds for all values of \texttt{X} and \texttt{Y}.

At runtime, a for-all is executed by calling \texttt{G} with \texttt{V} 
unbound. If \texttt{G} succeeds, the value of \texttt{V} is checked. If 
\texttt{V} is still unbound, \texttt{G} is added to the CHS and the for-all 
succeeds. If \texttt{V} is bound, failure and backtracking take place. However, 
if \texttt{V} is negatively constrained, \texttt{G} is added to the CHS and 
then called again for each value in \texttt{V}'s prohibited value list, 
substituting the constrained value for \texttt{V}. If \texttt{G} succeeds for 
every constrained value, the for-all succeeds and \texttt{G} is added to the 
CHS with \texttt{V} unbound. Consider the following program with dual rules 
added:
\begin{verbatim}
p :- not q(X).
q(Y) :- Y = a.
q(Y) :- Y \= a.

not p :- forall(X, np1(X)).
np1(X) :- q(X).

not q(Y) :- nq1(Y), nq2(Y).
nq1(Y) :- Y \= a.
nq2(Y) :- Y = a.
\end{verbatim}
\noindent and the query \texttt{?- not p.} The clause for \texttt{not p} will 
execute \texttt{forall(X, np1(X))}, which will in turn execute \texttt{q(X)}. 
The first clause of \texttt{q(Y)} will succeed, unifying \texttt{Y} with 
\texttt{a}. However, because \texttt{X} will be bound when \texttt{np1(X)} 
succeeds, execution will fail and backtrack, forcing the second clause for 
\texttt{q(Y)} to be selected. This clause will succeed, adding \texttt{a} to 
\texttt{Y}'s prohibited value list. Because \texttt{X} is now constrained when 
\texttt{np1(X)} succeeds, the forall will call \texttt{np1(X)} substituting 
each member of \texttt{X}'s prohibited value list for \texttt{X} in turn. In 
this case, the only extra call will be \texttt{np1(a)}. The call to 
\texttt{q(a)} will succeed via the first clause for \texttt{q(Y)}, causing the 
call \texttt{np1(a)} to succeed. Finally, since \texttt{np1(X)} has succeeded 
both with \texttt{X} constrained and for every member of \texttt{X}'s 
prohibited value list, the forall itself will succeed, setting \texttt{X}'s 
prohibited value list to an empty list before adding it to the CHS.

Two points are worth noting. First, for-alls are only used in internally 
generated code; they are not made available to users. This allows us to 
ensure that \texttt{V} will always be unbound at the time the call is made. 
Second, the correctness of our for-all algorithm relies on our use of the 
s(ASP) universe, and specifically on the properties given in Proposition 
\ref{th:varproperties}. The above algorithm will not work in cases where the 
domain of a variable may be finite or where it is possible to constrain a 
variable against every element of its domain. For example, given the rule
\begin{verbatim}
p(1).
\end{verbatim}
the goal \texttt{forall(X, p(X))} will fail under our method. This behavior 
is correct when the domain of \texttt{X} is the s(ASP) universe, as required by 
our method. However, were the domain of \texttt{X} allowed to be $\{ 1 \}$, 
the failure of the forall would be incorrect.

It is important to note that our use of dual rules may sometimes produce 
unexpected results. Because the disunification of two negatively constrained 
variables is illegal under our program restrictions (Section 
\ref{sec:restrictions}), it is implicitly illegal to call the dual of a rule 
which contains a non-ground structure in the head, unless the corresponding 
argument in the call is structured such that no comparisons between a variable 
and a non-ground element are made. Since the argument in the call and the 
structure in the head of the rule are unified in the original clause, the dual 
will contain a disunification between these elements. Consider the rule:
\begin{verbatim}
p([X | T]) :- q(X), p(T).
\end{verbatim}
\noindent Because X and T are present in a structure in the head, the first 
step in creating a dual would be to abstract them out of the head, as discussed 
earlier in this section. However, that makes them body variables, necessitating
the use of a forall. Thus, the final dual is:
\begin{verbatim}
not p(Z) :- forall(X, forall(T, np1(Z, X, T))).
np1(Z, X, T) :- Z \= [X | T].
np1(Z, X, T) :- Z = [X | T], not q(X).
np1(Z, X, T) :- Z = [X | T], q(X), not p(T).
\end{verbatim}
If this dual is called with a list containing unbound variables, the 
disunification operation in the first clause of \texttt{np1(Z, X, T)} may 
encounter illegal cases and trigger a fatal error. For instance, given the call 
\texttt{not p([A | B])}, the corresponding call will be \texttt{np1([A | B], X, 
T).} The disunification operation in the first clause of np1(Z, X, T) will 
attempt to recursively disunify \texttt{A} with \texttt{X} and \texttt{B} with 
\texttt{T}. As both operations are illegal, execution will halt and report an 
error. The same situation will occur if the same variable occurs more than once 
in the head of a rule. Because subsequent occurrences are replaced with unique 
variables and goals are added to the rule body unifying them with the original, 
the dual contains corresponding disunification operations. As a result, calling 
the dual with non-ground values for such variables will trigger an error. 
Consider the following rule and its dual:
\begin{verbatim}
t(A, A).
not t(A, B) :- A \= B.
\end{verbatim}
Because the original rule contains no disunification operations, \texttt{t(A, 
A)} can be called for any value of \texttt{A}. However, the dual will produce 
an error unless either \texttt{A} or \texttt{B} is bound such that no 
negatively constrained variables are disunified. Consider the following query:
\begin{verbatim}
?- not t(A,2), not t(B,1), not t(A,B).
\end{verbatim}
The call \texttt{not t(A, 2)} will succeed constraining \texttt{A} against 2 
and the call \texttt{not t(B, 1)} will succeed constraining \texttt{B} against 
1. However, the call \texttt{not t(A, B)} will trigger an error when it 
attempts to disunify \texttt{A} and \texttt{B}.


\subsubsection{Restrictions on Legal Programs} \label{sec:restrictions}

While one of our primary goals in developing s(ASP) has been to support the 
largest possible class of normal logic programs, some restrictions are still 
required under our current method. Now that we have introduced the other 
core concepts, we are equipped to discuss these restrictions.

\begin{definition} \label{def:restrictions}
	A \textbf{legal program} is a normal logic program, as defined 
	in Definition \ref{def:normal}, which satisfies the following restrictions:
	\begin{enumerate}
		\item Operands of arithmetic operations must be ground at the time they 
		are executed.
		\item Left recursion cannot lead to success.
		\item A negatively constrained variable (see Section 
		\ref{sec:variables}), cannot be \textit{disunified} with or 
		\textit{constrained against} another negatively constrained variable.
	\end{enumerate}
\end{definition}

Of these restrictions, we believe that the first two may eventually be lifted 
by modifying our method. Delayed expansion and improved constraint handling 
may remove the need for ground arithmetic, while tabling should eliminate 
issues related to left recursion. At present, these changes are left for future 
work.

The remaining restriction is integral to our current method. Recall from 
Proposition \ref{th:varproperties} that one of the key properties of our 
method is that the domain of a negatively constrained variable can never be 
empty with respect to the s(ASP) universe. Allowing a negatively constrained 
variable to be constrained against another with an empty prohibited value list 
would violate the property, breaking our method. Improved constraint 
handling might allow variable disunification to occur in cases where neither 
variable's domain would become empty, but we leave this for future work. For 
the time being, the only exception to this restriction is a special case 
involving even loops, described in Section \ref{sec:even}.

\subsection{Constructing the s(ASP) Method} \label{sec:execution}

Now that we have covered the necessary core concepts, we can discuss several 
other important aspects of the s(ASP) method. To avoid repetition, we will 
discuss these aspects in terms of differences between our propositional method, described in Section \ref{sec:grounded}, and our predicate method. Thus, in this section, we will look at changes made to 
coinduction (Section \ref{sec:cosld}), even loops (Section \ref{sec:even}) and 
consistency checking (Section \ref{sec:consistency}).

\subsubsection{Coinduction} \label{sec:cosld}
Like our propositional method, s(ASP) executes programs using a modified form 
of co-SLD resolution \cite{iclp2007}. The most obvious changes to the predicate method are the 
incorporation of negatively constrained variables and the corresponding use of 
constructive unification and disunification. However, additional changes are 
needed to handle negatively constrained variables correctly. In this section, 
we will look at the changes needed to adapt both coinductive failure and 
coinductive success for use with s(ASP).

Under our method for propositional programs, coinductive failure occurs when 
the negation of a goal is present in the CHS \cite{goalasp}. For example, if 
\texttt{not p} is found in the CHS when checking \texttt{p}, the call to 
\texttt{p} will fail coinductively. So, when dealing with propositional 
programs, it is sufficient to simply fail when a match for the negation is 
found. It might seem that extending this method to the predicate case would be 
as simple as checking to see if a goal's negation unifies with any entry in the 
CHS. However, this can lead to incorrect behavior when combined with CHS 
entries that include negatively constrained variables. Consider a program 
consisting of the following rule, with its dual added for convenience:
\begin{verbatim}
pi(X) :- X = 3.14.
not pi(X) :- X \= 3.14.
\end{verbatim}
Correct behavior requires that a call to \texttt{not pi(X)} should always 
succeed with \texttt{X} \textbackslash= 3.14. However, this may not happen 
if we rely on ordinary unification to check for coinductive failure. For 
instance, if we execute the program with the query \texttt{?- pi(Y), not 
pi(X)}, the goal \texttt{pi(Y)} will succeed for \texttt{Y} = 3.14 and 
execution will move on to \texttt{not pi(X)}. However, the negation of 
\texttt{not pi(X)} will unify with the CHS entry for \texttt{pi(Y)}, causing 
the query to incorrectly fail.

Our solution to the problem once again relies on constructive negation. Instead 
of simply failing when the negation of a goal unifies with an entry in the CHS, 
coinductive failure is viewed as a filter, allowing bindings whose negation 
does not unify with a CHS entry to succeed. This is accomplished by 
constraining variables in the call such that the call's negation will no longer 
unify with any entry in the CHS. In the above example, when \texttt{not pi(X)} 
is tested, X will be constrained against 3.14, preventing its negation, 
\texttt{pi(X)}, from unifying with the CHS entry for \texttt{pi(3.14)}.

Two things are worth noting about this process. First, the \textit{negation} of 
a goal is checked rather than its \textit{dual}. For example, the dual of 
\texttt{pi(3.14)} is \texttt{not pi(X)} where \texttt{X} is constrained 
against 3.14, but the negation of \texttt{pi(3.14)} is simply \texttt{not 
pi(3.14)}. Second, the coinductive failure check may be non-deterministic. 
Consider the following rule:
\begin{verbatim}
q(X) :- X \= 2, X \= 3.
\end{verbatim}
Given the query \texttt{?- q(X), not q(Y)}, at the time \texttt{not q(Y)} is 
called, the CHS will contain \texttt{q(X)}, with \texttt{X} constrained against 
both 2 and 3. This leaves two ways for the coinductive failure check to 
succeed: \texttt{Y} may be set to either 2 or 3. This choice will be made 
non-deterministically: 2 will be selected first, but 3 may be chosen when 
backtracking.

Non-determinism may also arise when testing goals with more than one argument. 
Consider a modification of the previous rule:
\begin{verbatim}
q(X, Y) :- X \= 2, Y \= 3.
\end{verbatim}
Now, given the query \texttt{?- q(X, Y), not q(A, B)}, at the time \texttt{not 
q(A, B)} is called, the CHS will contain \texttt{q(X, Y)}, with \texttt{X} 
\textbackslash= 2 and \texttt{Y} \textbackslash= 3. Again, this leaves two ways 
for the coinductive failure check to succeed: \texttt{A} may be set to 2 or 
\texttt{B} may be set to 3.

To ensure that all cases are considered when executing the coinductive failure 
check, each argument is considered separately, first to last, with subsequent 
arguments being selected on backtracking once all choices for the previous 
argument have been exhausted. Consider the rules:
\begin{verbatim}
q(W, X) :- W \= 2, W \= 3, X \= 2, X \= 3.
q(Y, 3) :- Y \= 2, Y \= 3.
\end{verbatim}
Now, given the query \texttt{?- q(W, X), q(Y, 3), not q(A, B)}, at the time 
\texttt{not q(A, B)} is called, the CHS will contain both \texttt{q(W, X)}, 
with \texttt{W} \textbackslash= 2, \texttt{W} \textbackslash= 3, \texttt{X} 
\textbackslash= 2, \texttt{X} \textbackslash= 3 and \texttt{q(Y, 3)}, with 
\texttt{Y} \textbackslash= 2, \texttt{Y} \textbackslash= 3. When executing the 
coinductive failure check for \texttt{not q(A, B)}, the first argument will be 
examined first, and the goal will be allowed to succeed first with \texttt{A} = 
2 and then with \texttt{A} = 3. Having exhausted all options for success via 
the first argument, further backtracking will lead to the second argument being 
examined. Here, success is possible only when \texttt{B} = 2. Thus, the 
coinductive failure check for \texttt{not q(A, B)} will succeed up to three 
times: once each for \texttt{A} = 2, \texttt{A} = 3 and \texttt{B} = 2. 
Together, these cases cover all possible scenarios where \texttt{not q(A, B)} 
does not unify with any element of the CHS.

Like coinductive failure, coinductive success must be also modified to work 
with predicate programs. In this case, we must differentiate between CHS 
entries and ancestors in the call stack which are exact matches for a call and 
those which simply unify with it.
\begin{definition}
	Two terms are an \textbf{exact match} if they can be constructively unified 
	without altering the prohibited value lists of any variables present in 
	either argument.
\end{definition}
For example, given variables \texttt{X} and \texttt{Y}, if both are constrained 
against 2, then they are an exact match. However, if \texttt{X} is constrained 
against 2 and \texttt{Y} against 2 and 3, they are unifiable, but not an exact 
match. The same applies to compound terms: \texttt{f(X)} and \texttt{f(Y)} are 
an exact match if \texttt{X} and \texttt{Y} are an exact match.

When testing for coinductive success, exact matches will allow success or 
failure to be deterministic, while other matches will be non-deterministic.
The testing process for a call \texttt{C} is as follows:
\begin{enumerate}
	\item \texttt{C} is tested against the CHS for exact matches. If one is 
	found, deterministically succeed.
	\item \texttt{C} is tested against each entry \texttt{D} in the call stack. 
	The number of negations between \texttt{C} and \texttt{D} are counted, 
	excluding \texttt{C} and \texttt{D} themselves.
	\begin{enumerate}
		\item If \texttt{C} and \texttt{D} are an exact match with no 
		intervening negations, fail deterministically.
		\item If \texttt{C} constructively unifies with \texttt{D} with an 
		even, non-zero number of intervening negations, succeed. If \texttt{C} 
		and \texttt{D} are an exact match, success is deterministic, otherwise 
		it is non-deterministic. This non-determinism simply allows \texttt{C} 
		to be executed normally by step 3 upon backtracking, allowing solutions 
		which might otherwise be missed. \textit{Unification here requires an 
		occurs check for correctness (see below).}
	\end{enumerate}
	\item If no matches are found or all deterministic matches have been 
	exhausted, execute \texttt{C} normally.
\end{enumerate}

The use of exact matches in steps 1 and 2(a) is necessary for completeness. 
With constructive unification, a call would always succeed if it unified an 
entry in the CHS and fail if it unified with an entry in the call stack with no 
intervening negations. However, this behavior could result in solutions being 
skipped. For example, in step 2(a), exact matches are needed to avoid 
false positives when detecting positive loops. Consider the rules
\begin{verbatim}
r(V) :- r(V2).
r(3.14).
\end{verbatim}
Were a program containing these rules to be grounded, positive loops would only 
be present for those cases where \texttt{V} = \texttt{V2}. However, since 
\texttt{V2} will always be unbound, \texttt{r(V)} and \texttt{r(V2)} will 
always constructively unify. As a result, a positive loop would always be 
detected if ordinary coinductive success were used, leading to failure in all 
cases. However, if exact matches are required, a positive loop will be detected 
only when \texttt{V} is unbound. As a result, the query \texttt{?- r(V)} will 
succeed for any \texttt{V}: a call with an unbound \texttt{V} will succeed via 
the second rule, while call with a ground or semi-ground \texttt{V} will 
succeed via the first, with \texttt{r(V2)} being satisfied by the second. 
Consider the query \texttt{?- r(1)}:
\begin{enumerate}
	\item \texttt{r(1)} will be checked for coinductive failure/success.
	\item Since the conditions for immediate success or failure are unmet, 
	\texttt{r(1)} will be added to the call stack and expanded using the 
	first rule.
	\item \texttt{r(V2)} will be checked for coinductive failure/success.
	\item Since the conditions immediate success or failure are unmet, 
	\texttt{r(V2)} will be added to the call stack and expanded using the 
	first rule.
	\item The new \texttt{r(V2)}, \texttt{r(V2')} for convenience, will be 
	checked for coinductive failure/success.
	\item Since \texttt{V2'} and \texttt{V2} both have empty prohibited value 
	lists, \texttt{r(V2')} and \texttt{r(V2)} are an exact match. As there 
	are no intervening negations between the two calls, \texttt{r(V2')} 
	will fail.
	\item Execution will backtrack to the expansion of \texttt{r(V2)} and 
	try the second rule.
	\item \texttt{r(V2)} will unify with \texttt{r(3.14)} and succeed.
	\item \texttt{r(1)} will succeed, returning the partial stable model 
	\{\texttt{r(1), r(3.14)}\}.
\end{enumerate}

The criteria for step 2(b) is similar to that of our propositional method:
coinductive success only occurs if an even, non-zero number of negations exist 
between a recursive call and its ancestor in the call stack \cite{goalasp}. 
However, when attempting to unify the current call with an ancestor in the call 
stack, an \textit{occurs check} is necessary to retain correctness. That is, a 
variable cannot unify with any term which contains the same variable. Consider 
the following rules for successor notation:
\begin{verbatim}
n(0).
n(s(X)) :- n(X).
\end{verbatim}
Given the query \texttt{?- n(s(s(X))).} s(ASP) will produce the partial stable 
models \{\texttt{n(0),n(s(0)),n(s(s(0)))}\}, 
\{\texttt{n(0),n(s(0)),n(s(s(0))),n(s(s(s(0))))}\} and so on. However, 
without an occurs check, the query would fail: \texttt{n(s(s(X)))} would be 
added to the call stack and \texttt{n(s(X))} would be called. However, when 
checking coinductive success, \texttt{n(s(X))} would unify with 
\texttt{n(s(s(X)))}, triggering failure due to the lack of intervening 
negations. The occurs check prevents the two calls from unifying, thereby 
preserving correctness.

It is also important to note that coinduction will prevent infinite looping 
unless unbounded recursion occurs in such a way that no recursive call can 
constructively unify with any of its ancestors. For instance, the following 
program will produce infinite looping:
\begin{verbatim}
s(X) :- X2 is X + 1, s(X2).
?- s(1).
\end{verbatim}
In this example, every instance of \texttt{s(X)} will be ground and unique, 
thus preventing coinductive success from ever occurring. However, such cases 
can be avoided by rewriting the rules in question to ensure that either a 
recursive call will eventually unify with some ancestor or that some base case 
will eventually be reached.

\subsubsection{Even Loops} \label{sec:even}


Even loops (Definition \ref{def:even}) have special significance in 
the stable model semantics: they indicate that a goal may be either 
true or false. Under s(ASP), a special case of even loops must be addressed for 
the sake of completeness: even loops containing loop variables.
\begin{definition}
	\textbf{Loop variables} are variables which occur in both a recursive call 
	and its ancestor in an even loop, and are unbound or negatively constrained 
	when the ancestor call succeeds.
\end{definition}

Observe that for a variable to be present in both a recursive 
call and its ancestor, it must also occur in each of the intervening calls 
which are part of the even loop. In this situation, the chain of literals in 
the loop may be either true or false for every grounding of the loop variables 
which does not produce a contradiction in the CHS. Because the s(ASP) universe 
is infinite, the result is that any program with at least one stable model 
containing a loop variable will have an infinite number of stable models which
contain an infinite number of elements, and thus have an infinite number of 
partial stable models. Consider the following program:
\begin{verbatim}
p(X, Y) :- not q(X, Y), t(Y, Y).
q(X, Y) :- not p(X, Y).
\end{verbatim}
with the query \texttt{?- q(X, Y)}. Because \texttt{X} and \texttt{Y} occur as loop 
variables, \texttt{q(X, Y)} and \texttt{not p(X, Y)} may be true or false for 
each possible grounding of \texttt{X} and \texttt{Y}, so long as their truth 
values are the same for a given combination (opposing truth values would 
produce a contradiction). For example, \{\texttt{q(1, 2), not p(1, 2)}\} and 
\{\texttt{p(a, b), q(1, 2), not p(1, 2), not q(a, b)}\} are both partial stable 
models of the above program. 
Furthermore, we know from the variable properties given in Proposition 
\ref{th:varproperties} that \texttt{X} and \texttt{Y}'s domains are infinite, 
so there must be an infinite number of these partial stable models.

With this in mind, completeness requires that we have a means of representing 
the potentially infinite number of partial stable models which may result from 
loop variables. Our mechanism for this is to prefix loop variables with a 
question mark (?) when printing them, indicating that literals containing them 
may be either true or false for each grounding of the loop variable which does 
not produce a contradiction. Observe that no information will be lost: it will 
always be possible to ground our output such that it will produce a subset of 
any given stable model for which the query would succeed. 
For example, for the query \texttt{?- q(X, Y)}, the above program will produce 
a single result, \{\texttt{q(?X, ?Y), not p(?X, ?Y)}\}, compactly representing 
an infinite number of partial stable models for each full stable model for 
which the query succeeds. Note that a loop variable will become an ordinary 
negatively constrained variable if used in a forall: having succeeded for all 
values, it can no longer be assigned false for any of them.

Finally, as mentioned in Section \ref{sec:unification}, loop variables force us 
to make an exception to the prohibition that two negatively constrained 
variables cannot be constrained against each other. Ordinarily, goals in the 
CHS cannot be further modified, and this includes placing further constraints 
upon any of their variables. However, in the case of loop variables, it is 
necessary to ensure that no grounding allowed by the output will be incorrect. 
Therefore, if a call \texttt{C} is added to the CHS with a loop variable and a 
call for \texttt{not C} later succeeds, correctness requires that the loop 
variable in \texttt{C}'s CHS entry be constrained against the corresponding 
entry in \texttt{not C}'s entry. If the domain of a loop variable is empty when 
the query and NMR check succeed, failure and backtracking must occur to ensure 
correctness.

\subsubsection{Consistency Checking} \label{sec:consistency}

Recall from Section \ref{sec:grounded} that any constraints imposed by OLON 
rules in a program are enforced by appending a special rule, the NMR check, to 
each query. The NMR check calls sub-checks for each OLON rule in the program, 
enforcing the constraints which they impose. Two changes are necessary to 
support predicate programs. First, the NMR sub-checks must be generated using 
our new dual rule algorithm. Second, because OLON rules apply global 
constraints, we must ensure that the sub-checks are satisfied for all possible 
values of their variables.

The first step is to use our new dual rule algorithm, detailed in Section 
\ref{sec:conneg}, when generating NMR sub-checks. As explained in Section 
\ref{sec:grounded}, a sub-check is created by first appending the negation of 
an OLON rule's head to its body (if not already present), taking the dual of 
the modified rule and assigning it a unique head. Thus, with our new dual rule 
algorithm, the rule
\begin{verbatim}
p(X) :- q(X), not p(X).
\end{verbatim}
will produce the sub-check
\begin{verbatim}
chk_p(X) :- not q(X).
chk_p(X) :- q(X), p(X).
\end{verbatim}

The second issue is that the constraints imposed by OLON rules are global, so 
the resulting sub-checks must hold for all possible bindings of their 
variables. This is accomplished using our for-all mechanism, also described in 
Section \ref{sec:conneg}. The body of the NMR check is modified by considering 
each variable in a sub-check goal to be universally quantified, and abstracting 
it with a forall. Thus, the NMR check for the above sub-check would be
\begin{verbatim}
nmr_check :- forall(X, chk_p(X)).
\end{verbatim}
In the case of headless rules, any variables are already body variables, so the 
necessary forall wrappers will be applied when the sub-checks are created.

With these modifications, the NMR check will correctly apply any constraints 
imposed upon the program. However, this does make it more difficult to identify 
whether a program is legal without running it. Because each sub-check is a dual 
rule, the same caveats which apply to calling negated goals, discussed in
Section \ref{sec:conneg}, also apply to OLON rules in general.

It is interesting to note that one aspect of our method can remain almost 
unchanged: the initial detection of OLON rules. As explained in Section 
\ref{sec:grounded}, OLON rules are detected by finding cycles in the call 
graph which contain an odd number of negations. At a glance, the addition of 
variables would appear to complicate this procedure. Consider the following 
rule:
\begin{verbatim}
p(X) :- q(X, Y), not p(Y).
\end{verbatim}
\noindent An odd loop is present, but only when \texttt{X $=$ Y}. It is easy to 
assume that such a case might require a constraint to be added so that the 
resulting sub-check will exclude cases where \texttt{X} \textbackslash= 
\texttt{Y}, but this is actually unnecessary. In fact, variables and 
constraints can be ignored entirely when detecting OLON rules. Thus, the above 
rule will produce the following NMR check and sub-checks:
\begin{verbatim}
nmr_check :- forall(X, chk_p(X)).

chk_p(X) :- forall(Y, chk_p2(X,Y)).

chk_p2(X,Y) :- not q(X,Y).
chk_p2(X,Y) :- q(X,Y), p(Y).
chk_p2(X,Y) :- q(X,Y), not p(Y), p(X).
\end{verbatim}
\noindent Observe that when \texttt{X} \textbackslash= \texttt{Y}, the 
sub-check will always succeed: one of the first two clauses for 
\texttt{chk\_p2(X, Y)} will succeed in cases where the original rule for 
\texttt{p(X)} would fail and the third clause will succeed in cases where the 
original rule would succeed. Indeed, this will always be the case for such 
``conditional'' OLON rules: in cases where no OLON is present, the 
corresponding sub-check will always be satisfied. Therefore, while adding 
constraints to exclude non-OLON cases might improve performance, they are not 
required for correctness.

\subsection{Overview of the Completed Method} \label{sec:completed}

Now that we have looked at the individual components, we can examine the 
method as a whole.

Under the \textbf{s(ASP) method}, a legal program \texttt{P} is executed 
with a query \texttt{Q} as follows: First, the call graph of \texttt{P} is 
examined to create the NMR check and sub-checks. Next, the body of the NMR 
check is appended to \texttt{Q}. Then, each goal \texttt{G} in \texttt{Q} 
is executed in order. If \texttt{G} is an arithmetic expression, a unification 
or disunification operation, or a forall, it is executed accordingly. 
Otherwise, \texttt{G} is checked against the CHS:
\begin{itemize}
	\item If the CHS contains an exact match for \texttt{not G}, \texttt{G} 
	fails \textit{deterministically}.
	\item If the CHS contains an exact match for \texttt{G}, \texttt{G} 
	succeeds \textit{deterministically}.
	\item If no exact match is present in the CHS, \texttt{G} is 
	constrained against any CHS entries which unify with \texttt{not G}. 
	This process may be \textit{non-deterministic}.
\end{itemize}
Should \texttt{G} pass the CHS check without succeeding or failing, the call 
stack is examined for any cycles containing \texttt{G}, starting at the 
most recent call and working back:
\begin{itemize}
	\item If \texttt{G} is an exact match for an entry in the call stack 
	with no intervening negations, \texttt{G} fails 
	\textit{deterministically} (positive loop).
	\item If \texttt{G} is an exact match for an entry in the call stack with 
	an even, non-zero number of negations, \texttt{G} succeeds 
	\textit{deterministically} (coinductive success).
	\item If \texttt{G} \textit{constructively unifies} with an entry in the 
	call stack with an even, non-zero number of negations, \texttt{G} succeeds 
	\textit{non-deterministically} (coinductive success).
\end{itemize}
Should \texttt{G} pass the call stack check without succeeding or failing, 
either \texttt{G} does not match any entry in the call stack or all 
non-deterministic options have been exhausted. Then, \texttt{G} is expanded 
using the rules in \texttt{P}. If \texttt{G} succeeds in this manner, it is 
added to the CHS. If every goal in \texttt{Q} succeeds, the query succeeds. 
Finally, the domains of any loop variables are checked to ensure that they are 
non-empty. If so, execution succeeds and the elements in the CHS form a partial 
stable model of \texttt{P}. Figure \ref{fig:metainterpreter} contains an 
abstract meta-interpreter for this method.

\begin{figure}
	\begin{center}
		\begin{verbatim}
		sasp(Q, NMR) :-
		    append(Q, NMR, Q2),
		    exec_goals(Q2),
		    check_loop_variable_domains,
		    print_chs.
		
		exec_goals([X | T]) :-
		    exec_goal(X),
		    exec_goals(T).
		exec_goals([]).
		
		exec_goal(X) :-
		    X = forall(V, G), !,
		    exec_forall(V, G).
		exec_goal(X) :-
		    check_chs_and_call_stack(X, CHSResult),
		    exec_goal2(X, CHSResult).
		
		exec_goal2(X, success). % coinductive success
		exec_goal2(X, expand) :- % expand using rules
		    get_matching_rule(X, R), % select subsequent rules on backtracking
		    exec_goals(body(R)),
		    add_to_chs(X).
		
		exec_forall(V, G) :-
		    unbound(V), % fail if variable is bound or constrained
		    exec_goal(G), % first solve goal normally
		    Cons = get_constraints(V), % fail if variable is bound
		    exec_with_each_constraint_value(G, V, Cons),
		    set_unbound(V), % goal succeeded for all V.
		    add_to_chs(G).
		
		check_chs_and_call_stack(X, failure) :-
		    Xn = dual(X),
		    exact_match_in_chs(Xn), !. % coinductive failure unavoidable
		check_chs_and_call_stack(X, success) :-
		    exact_match_in_chs(X), !. % coinductive success unavoidable
		check_chs_and_call_stack(X, Result) :- % avoid failure, if possible.
		    constrain_goal_against_unifiable_duals(X), % non-deterministic.
		    check_call_stack(X, Result).
		
		% Check the call stack for cycles on current goal. If a cycle over an
		% exact match is found, don't look for other matches when backtracking.
		check_call_stack(X, failure) :-
		    call_stack_has_positive_cycle_w_exact_match(X), !.
		check_call_stack(X, success) :-
		    call_stack_has_even_cycle_w_exact_match(X), !.
		check_call_stack(X, success) :-
		    call_stack_has_even_cycle(X), % non-deterministic
		    unify_goal_w_match.
		check_call_stack(X, expand).
		\end{verbatim}
	\end{center}
	\caption{Abstracted s(ASP) Meta-interpreter}
	\label{fig:metainterpreter}
\end{figure}

\section{Correctness} \label{sec:correct}

In this section, we will discuss the correctness of the s(ASP) method. We 
will show that s(ASP) is sound for all legal programs (Definition 
\ref{def:restrictions}) and argue that, while completeness is in fact 
impossible, the method is still useful for the vast majority of practical 
programs. We will begin by looking at the propositional case and then move on 
to the predicate case.

\subsection{Propositional Programs} \label{sec:groundcorrect}

In the proofs below, we demonstrate that s(ASP) is sound for all legal programs 
grounded over the s(ASP) universe and complete for the all finite, ground, 
legal programs. Note that this completeness class includes finitely groundable 
programs which have been grounded over the Herbrand universe prior to execution.

\begin{theorem} \label{th:groundequiv}
	For legal, finite, ground programs, our predicate method is equivalent 
	to our propositional method.
\end{theorem}

\begin{proof}
	This theorem holds by design. Our propositional method forms the basis for 
	our predicate method, and all the modifications we have introduced 
	apply only to cases involving variables and illegal programs. Thus, for a 
	legal, finite, ground program, which, by definition, contains no variables, 
	the two methods will be equivalent.
\end{proof}

\begin{corollary}
	Our predicate method is sound and complete for finite, legal, ground
	programs.
\end{corollary}

\begin{proof}
	Our ground method has already been proven sound and complete for the set 
	of all finite programs grounded over the Herbrand universe \cite{goalasp}. 
	Therefore, by Theorem \ref{th:groundequiv}, our predicate method is 
	also sound and complete for these programs.
\end{proof}

\begin{theorem} \label{th:groundsound}
	Our predicate method is sound for legal programs grounded over the 
	s(ASP) universe.
\end{theorem}

\begin{proof}
	Our propositional method was originally proven sound with respect to the 
	original GL method \cite{goalasp}. It is a simple matter to extend the 
	original soundness proof to show the equivalence of our propositional 
	method 	to the modified GL method described in Section 
	\ref{sec:universe}, thereby proving it sound for legal programs grounded 
	over the s(ASP) universe. Thus, by Theorem \ref{th:groundequiv}, our 
	predicate method is also sound for legal programs grounded over the s(ASP) 
	universe.
\end{proof}

\subsection{Predicate Programs}

The case of predicate logic programs is significantly more complicated than the 
propositional case. While s(ASP) is sound for all legal programs, completeness 
over this class is impossible, as we will discuss.

\subsubsection{Soundness}

Informally, the s(ASP) algorithm is sound if every partial model it generates 
is part of some stable model of the program. To show this we must first 
separate the (grounded) program into three parts:

\begin{itemize}
	\item The set of clauses required to prove the partial model.
	\item The set of clauses related to the query but not needed to prove the 
	partial model.
	\item The set of clauses not related to the query.
\end{itemize}

We will first show that the second set of clauses can be removed without 
affecting the result. That is, the partial model is a part of some stable model 
of the new program, and that stable model is a stable model of the original 
program. We also show that clauses needed to prove that some literal is false 
can be modified by removing goals as long as one ``false'' goal remains. Using 
these two modifications and the splitting theorem \cite{Lifschitz94}, we can 
isolate a subprogram comprised only of clauses and literals touched by the 
s(ASP) algorithm.

Using the terminology of \cite{Lifschitz94}, detailed in the next 
section, this subprogram can be considered the ``bottom'' of the program, 
with the remaining clauses forming the ``top'' of the program. We show that the 
partial model generated by the s(ASP) algorithm is a stable model of the bottom 
of the program. Then, in accordance with the splitting theorem we transform the 
top of the program and show that since the NMR check is satisfied a stable 
model exists for it. By using the splitting theorem and the points discussed 
above, we show that the union of the partial model and the top's stable model 
is a stable model the original program.

Before continuing, we define the following for convenience:
\begin{definition}
	Let \texttt{G} be a goal constructed from atom \texttt{A}. Then, 
	\texttt{atom(G)} = \texttt{A}, and \texttt{goal(A)} is the set of both 
	goals (positive and negative) that can be constructed from \texttt{A}. The 
	arguments of \texttt{A} are said to be the arguments of \texttt{G}.
\end{definition}

\begin{definition}
	Let \texttt{R} be a clause of the form:\\
	\texttt{
		\indent p :- q$_1$, ..., q$i$, ..., q$_m$,\\
		\indent~~~~~not r$_1$, ..., not r$_j$, ..., not r$_n$.\\
	}
	Then:
	\begin{itemize}
		\item \textbf{head(R)} = \texttt{p}
		\item \textbf{pos(R)} = \texttt{$\{$ q$_1$, ..., q$i$, ..., q$_m$ $\}$}
		\item \textbf{neg(R)} = \texttt{$\{$ not r$_1$, ..., not r$_j$, ..., 
			not r$_n$ $\}$}
		\item \textbf{lit(R)} = \texttt{$\{$ H $\}$ $\cup$} pos(\texttt{R}) 
		\texttt{$\cup$} neg(\texttt{R})
	\end{itemize}
\end{definition}

\paragraph{Review of Splitting Theorem}

Next, we review the splitting theorem \cite{Lifschitz94}. A set of literals can 
be used to \textit{split} a ground program. This set is called a 
\textit{splitting set}, and is defined as follows.

\begin{definition}
	A set of ground atoms \texttt{U} is a \textbf{splitting set} for some 
	ground program \texttt{P} if for every clause \texttt{R} in \texttt{P}, 
	head(\texttt{R}) $\in$ \texttt{U} $\Rightarrow$ lit(\texttt{R}) $\subseteq$ 
	\texttt{U}.
\end{definition}

The program is divided into two parts, the top and the bottom. The bottom is 
the set of clauses related to the splitting set and the top is the set of all 
other clauses.
\begin{definition}
	Let \texttt{P} be a ground program and \texttt{U} a splitting set for \texttt{P}. The \textbf{bottom} of \texttt{P} with respect to \texttt{U}, specified as \texttt{b$_U$}(\texttt{P}), is the set of clauses \texttt{r} for which lit(\texttt{r}) $\subseteq$ \texttt{U}. The set \texttt{P} $\setminus$  \texttt{b$_U$}(\texttt{P}) is the \textbf{top} of \texttt{P} with respect to \texttt{U}.
\end{definition}

The stable models of a ground program \texttt{P} can be computed by combining the stable models of \texttt{b$_U$}(\texttt{P}) and the stable models generated by the top. This requires us to generate a new program from \texttt{P} $\setminus$  \texttt{b$_U$}(\texttt{P}) based on the stable models for \texttt{b$_U$}(\texttt{P}).

\begin{theorem}
	Let \texttt{P} be some ground program, \texttt{U} a splitting set of \texttt{P}, and \texttt{X} a stable model for \texttt{b$_U$}(\texttt{P}). For each clause \texttt{r} in \texttt{P} such that pos(\texttt{r}) $\subset$ \texttt{X} and neg(\texttt{r}) $\cap$ \texttt{X} $= \emptyset$, we define a new clause \texttt{r}$^\prime$ with:
	\begin{itemize}
		\item head(\texttt{r}$^\prime$) = head(\texttt{r}),
		\item pos(\texttt{r}$^\prime$) = pos(\texttt{r}) $\setminus$ \texttt{U}
		\item neg(\texttt{r}$^\prime$) = neg(\texttt{r}) $\setminus$ \texttt{U}
	\end{itemize}
	
	We define the program \texttt{e}$_U$(\texttt{P} $\setminus$ \texttt{b}$_U$(\texttt{P}),\texttt{ X}) as the set of of all such new clauses, and for some stable model \texttt{Y} of \texttt{e}$_U$(\texttt{P} $\setminus$ \texttt{b}$_U$(\texttt{P}), \texttt{X}), \texttt{X} $\cup$ \texttt{Y} is a stable model of \texttt{P}.
	
	If either \texttt{b$_U$}(\texttt{P}) or \texttt{e}$_U$(\texttt{P} $\setminus$ \texttt{b}$_U$(\texttt{P}), \texttt{X}) has no stable model then there is no stable model for \texttt{P}.
\end{theorem}

\begin{proof}
	Proofs for these results are available in the original paper which 
	introduced the splitting theorem \cite{Lifschitz94}.
\end{proof}

\paragraph{Stripping Unneeded Rules and Body Literals}

\begin{lemma}\label{unneeded_clauses}
	Let \texttt{P} be a ground program, and \texttt{M} a stable model of \texttt{P}. Let \texttt{R} be a clause not in \texttt{P} such that the head of \texttt{R} is in \texttt{M}. Then, \texttt{M} is a stable model of the program \texttt{P} $\cup$ $\{$ \texttt{R} $\}$. 
\end{lemma}

\begin{proof}
	There are two cases for \texttt{R}:
	\begin{description}
		\item[Case 1:] There exists \texttt{L} $\in$ neg(\texttt{R}) such that \texttt{L} $\in$ \texttt{M}. In this case \texttt{R} is removed when computing the reduct, and the reduct does not change. Therefore, \texttt{M} is the least model of the reduct.
		\item[Case 2:] \texttt{R} is not removed when the reduct is computed. We will call the transformed clause \texttt{R}$^\prime$.
		
		Since the head of \texttt{R} is in \texttt{M} there must exist some clause R$_2$ in \texttt{P} with head(\texttt{R}) = head(R$_2$) such that for all \texttt{L} $\in$ pos(R$_2$), \texttt{L} $\in$ \texttt{M} and for all \texttt{L} $\in$ neg(R$_2$), \texttt{L} $\not\in$ \texttt{M}. We will call the transformed clause in the reduct \texttt{R}$^{\prime\prime}$.
		
		The only way \texttt{R}$^\prime$ can affect the least model of \texttt{P} $\cup$ $\{$ R $\}$ is to be used to place its head in it. Therefore, since \texttt{R}$^\prime$ can not affect literals besides its head, the body literals in \texttt{R}$^{\prime\prime}$ must be in the least model of the reduct of \texttt{P} $\cup$ $\{$ R $\}$. Thus, head(\texttt{R}$^{\prime\prime}$) (which is also head(R$^\prime$)) must also be in the least model, and the least model for the reduct of \texttt{P} $\cup$ $\{$ R $\}$ is the same as for \texttt{P}'s reduct. So, \texttt{M} must be the least model of the reduct of \texttt{P} $\cup$ $\{$ R $\}$.
	\end{description}
	
	Thus, \texttt{M} is a stable model of \texttt{P} $\cup$ $\{$ R $\}$.
	
\end{proof}

\begin{lemma}\label{unneeded_goals}
	Let \texttt{P} be a ground program and \texttt{M} a stable model of \texttt{P}. Let \texttt{R} be a clause in \texttt{P} such that there exists a literal in the body that is not in \texttt{M}, and let \texttt{G} be a ground goal. Let \texttt{P}$^\prime$ be the program constructed by adding \texttt{G} to the body of \texttt{R}. \texttt{M} is a stable model of \texttt{P}$^\prime$.
\end{lemma}

\begin{proof}
	Let \texttt{P} be a ground program, and \texttt{M} be a stable model of \texttt{P}. Let \texttt{R} be a clause in \texttt{P} such that the head of \texttt{R} is not in \texttt{M}. Let \texttt{G} be some ground goal.
	
	Create new program \texttt{P}$^\prime$ by adding \texttt{G} to the body of \texttt{R}. Call this clause \texttt{R}$^\prime$.
	
	We have two cases:
	\begin{enumerate}
		\item There exists some literal \texttt{L} $\in$ neg(\texttt{R}) and \texttt{L} $\in$ \texttt{M}, or
		\item there is some literal \texttt{L} $\in$ pos(\texttt{R}) such that \texttt{L} $\not\in$ \texttt{M}.
	\end{enumerate}
	
	In case 1, we know that \texttt{R}$^\prime$ will be removed when computing the reduct, and thus can not affect the least model. So we only need to consider case 2. In addition we can assume that \texttt{G} does not cause the clause to be removed from the reduct(otherwise it can not affect the least model). Now, notice that the addition of \texttt{G} does not affect the truth value of \texttt{L}. Thus, it is possible that the same process that causes \texttt{L} to not be in the least model of the reduct of \texttt{P} will cause it to not be in the reduct of \texttt{P}$^\prime$. So, \texttt{R}$^\prime$ cannot be used to place its head in the least model. Since no other clauses have changed, the least model of the reduct of \texttt{P}$^\prime$ is the same as the one for \texttt{P}, and thus \texttt{M} is a stable model of \texttt{P}$^\prime$.
\end{proof}

To prove theorem \ref{model_subset}, we want to separate the part needed to prove the query from the rest. Since the s(ASP) algorithm is goal directed we need to remove clauses related to the query that are not needed to prove it and goals in clauses related to the query, but not needed to prove it. We call this process \textbf{trimming}. Then we will make use the splitting theorem from \cite{Lifschitz94} to divide the new program into two parts, treating the portion needed to prove the query as the bottom, and the rest as the top.

It is important to remember that the s(ASP) algorithm works directly with the ungrounded program, but we will be trimming its ground program. When executing a clause (using it to prove some goal) it is possible (and likely) that variables in the clause will be constrained or bound by some goals in its body. If we associate with a clause a function for the domains of the variables in it, we can treat the state of the clause before and after as separate clauses.

\begin{definition}
	Let \texttt{R} be a clause.
	
	Let $\delta$(\texttt{R}, \texttt{X}) be the domain (set of possible groundings) of the variable \texttt{X} in \texttt{R}. \texttt{X} may be bound or constrained (though the constraint list may be empty). A clause in a program before execution will always have all variables unconstrained and therefore they will have the s(ASP) universe as their domains.
	
	The result of the execution of \texttt{R} succeeding is called $\sigma$(R). The variables in $\sigma$(\texttt{R}) are the same as in \texttt{R}, and for all variables \texttt{X} in \texttt{R}, $\delta$($\sigma$(\texttt{R}), \texttt{X}) $\subseteq$ $\delta$(\texttt{R}, \texttt{X}).
\end{definition}

\begin{definition}
	To \textbf{trim} a program we will follow the following algorithm. Let \texttt{P} be a program, \texttt{P}$^\prime$ be the result of grounding \texttt{P} over the s(ASP) universe, \texttt{M} a partial model of \texttt{P} generated by the s(ASP) algorithm, and \texttt{M}$^\prime$ the grounding of \texttt{M} over the s(ASP) universe. Let $\Phi$(\texttt{R}) be the set of clauses in \texttt{P}$^\prime$ generated by grounding clause \texttt{R} in \texttt{P}.
	\begin{enumerate}
		\item While \texttt{M} is begin computed: Let \texttt{G} be the current goal such that \texttt{G} is not a negated goal or a built-in/system generated literal such as the NMR check. Let \texttt{R} be the clause that is selected. If the execution of \texttt{R} succeeds, then mark all clauses in \texttt{P}$^\prime$ that are in $\Phi$(\texttt{R}) and could be considered a grounding of $\sigma$(\texttt{R}).
		\item After \texttt{M} is computed: Create a new program P$^{\prime\prime}$ from \texttt{P}$^\prime$ by:
		\begin{itemize}
			\item removing all clauses \texttt{R} from \texttt{P}$^\prime$ for which the head(\texttt{R}) $\not\in$ \texttt{M}$^\prime$ and \texttt{R} is not marked, and
			\item transform all clauses \texttt{R} in \texttt{P}$^\prime$ for which not head(\texttt{R}) $\in$ \texttt{M}$^\prime$ by removing all body goals for which nether they nor their negations are in \texttt{M}.
		\end{itemize}
		\item P$^{\prime\prime}$ is the result of \textbf{trimming} \texttt{P}$^\prime$ with respect to \texttt{M}$^\prime$.
	\end{enumerate}
\end{definition}

\paragraph{Forall}

\begin{lemma}\label{forall}
	Let \texttt{P} be a s(ASP) program, \texttt{G} be an atom, and \texttt{X} an unconstrained variable in \texttt{G}. Let $\mathcal{G}$ be the set of all goals obtained by grounding \texttt{X} in \texttt{G} over the s(ASP) universe. Then, a \texttt{forall(\texttt{X}, \texttt{G})} in \texttt{P} succeeds if and only if all goals in $\mathcal{G}$ succeed.
\end{lemma}

\begin{proof}
	Assume the opposite is true. That is, either \texttt{forall(\texttt{X}, \texttt{G})} succeeds and some \texttt{L} $\in$ $\mathcal{G}$ fails, or all goals in $\mathcal{G}$ succeed, but \texttt{forall(\texttt{X}, \texttt{G})} fails.
	\begin{description}
		\item[Case 1:] Suppose \texttt{forall(\texttt{X}, \texttt{G})} succeeds, but there exists some \texttt{L} $\in$ $\mathcal{G}$ such that \texttt{L} fails. There are two phases for the forall to succeed. First, \texttt{G} must succeed with \texttt{X} unbound. Then prove \texttt{G} with \texttt{X} grounded with each of its constraints. Since \texttt{L} fails, the value corresponding to \texttt{X} in \texttt{L} could not have been in the constraint list. Otherwise, the forall would have failed in the second phase. However, we could take the proof tree generated by the first phase, and ground \texttt{X} to obtain a proof tree for \texttt{L}, meaning there is a way for \texttt{L} to succeed. A contradiction.
		\item[Case 2:] Suppose \texttt{forall(\texttt{X}, \texttt{G})} fails, but all \texttt{L} $\in$ $\mathcal{G}$ succeed. We know that the forall could not have failed in the second phase, otherwise there would be some \texttt{L} $\in$ $\mathcal{G}$ such that \texttt{L} fails. Therefore there are two possibilities. Either, there is no way for \texttt{G} to succeed or all ways require \texttt{X} to be ground. The second case cannot be the case since all \texttt{L} $\in$ $\mathcal{G}$ succeed, and by definition, the s(ASP) universe contains an infinite number of terms that do not appear in the herbrand universe of \texttt{P}. Thus there exists some term in the s(ASP) universe for which \texttt{X} cannot be explicitly grounded against. The first case also cannot be true since all \texttt{L} $\in$ $\mathcal{G}$ succeed, thus there must be a way for \texttt{G} to succeed. A contradiction.
	\end{description}
	
	Therefore, \texttt{forall(\texttt{X}, \texttt{G})} succeeds if and only if all \texttt{L} $\in$ $\mathcal{G}$ succeeds.
\end{proof}

\paragraph{Constructive Coinductive Failure}

\begin{lemma}\label{ccf}
	Let \texttt{P} be a program, \texttt{G} a goal currently in the CHS, and \texttt{G}$^\prime$ be a goal that unifies with the negation of \texttt{G} that we wish to prove. Let \texttt{G}$^{\prime\prime}$ be the goal generated by the constructive coinductive failure algorithm from \texttt{G} and \texttt{G}$^\prime$. Then, \texttt{G} does not unify with \texttt{G}$^{\prime\prime}$.
\end{lemma}

\begin{proof}
	Without loss of generality, assume \texttt{G} is negated. So, \texttt{G}$^\prime$ will not be.
	
	If we assume that given an argument in \texttt{G}$^\prime$ we can restrict it so that it does not unify with the corresponding argument in \texttt{G}, we can easily see that if the algorithm succeeds \texttt{G}$^{\prime\prime}$ cannot unify with atom(\texttt{G}). So, we must show that if we restrict an argument (the algorithm succeeds for the argument) in \texttt{G}$^\prime$ it is never the case that it unifies with the corresponding argument in \texttt{G}. 
	
	We will show this by inducting over the depth of the term. The depth of a variable or constant is zero, and the depth of a list or function is one more than the maximal depth of all its arguments. Let \texttt{T}$^\prime$ be the argument from \texttt{G}$^\prime$ and \texttt{T} the corresponding argument from \texttt{G}. Notice that \texttt{T} and \texttt{T}$^\prime$ must unify. If \texttt{T} is a loop variable, then \texttt{T}$^\prime$ will be added to its constraint list, even if \texttt{T}$^\prime$ is a variable. In this case it is obvious that \texttt{T} and \texttt{T}$^\prime$ no longer unify. If \texttt{T} is an unconstrained variable and not a loop variable, the algorithm will always fail(since it unifies with everything), so we will not explicitly consider this case below.
	
	\begin{description}
		\item[Base Case:] Suppose \texttt{T}$^\prime$ has a depth of zero. If it is a constant the algorithm fails, so we only need to consider the case it is a variable. The behavior of the algorithm depends on what \texttt{T} is.
		\begin{itemize}
			\item If \texttt{T} is a constant, then \texttt{T}$^\prime$ is constrained against it. It is obvious in this case that they do not unify.
			\item If \texttt{T} is a constrained variable, then a term is nondeterministically selected from its constraint list for which \texttt{T}$^\prime$ is not constrained against, and is used to ground \texttt{T}$^\prime$. Again, it is obvious they do not unify.
			\item If \texttt{T} is a list or function, \texttt{T} will be added to the constraint list for \texttt{T}$^\prime$, and therefore no longer unify with \texttt{T}$^\prime$.
		\end{itemize}
		\item[Inductive Hypothesis:] Let \texttt{T}$_2^\prime$ and \texttt{T}$_2$ be terms that unify, with \texttt{T}$_2^\prime$ having a depth less than or equal to \texttt{k} and \texttt{T}$_2$ being from the goal in the CHS. Assume that if the algorithm succeeds the goal generated will not unify with \texttt{T}$_2$.
		\item[Inductive Step:] Suppose \texttt{T}$^\prime$ has a depth of \texttt{k} + 1. Then, \texttt{T}$^\prime$ must be a list or function, and since \texttt{T} unifies with \texttt{T}$^\prime$ it must also be a list or a function with the same functor and arity. So, we nondeterministically select an argument in \texttt{T}$^\prime$ and apply the algorithm to it with the corresponding argument in \texttt{T}. That argument must have depth of \texttt{k} or less, and by the inductive hypothesis if the algorithm succeeds then it cannot unify with the corresponding argument in \texttt{T}, and therefore by replacing the argument in \texttt{T}$^\prime$ with the result, we know that the new term cannot unify with \texttt{T}. If no argument can succeed then the algorithm will fail, so such a case can be ignored.
	\end{description}
	
	Thus, by induction if the algorithm succeeds the resulting term will not unify with the term in the CHS, and therefore \texttt{G}$^{\prime\prime}$ will never unify with atom(\texttt{G}).
\end{proof}

\paragraph{Soundness Theorem}

\begin{theorem} \label{model_subset}
	Let \texttt{P} be a program, and \texttt{M} a partial model of \texttt{P} generated by the s(ASP) algorithm. Let \texttt{M}$_2$ and \texttt{P}$_2$ be the results of grounding \texttt{M} and \texttt{P}, respectively, over the s(ASP) universe. There exists a stable model \texttt{X} of \texttt{P}$_2$ such that for all literals \texttt{L} in \texttt{M}$_2$, \texttt{L} is in \texttt{X}, and for all literals \texttt{L} with \texttt{not L} in \texttt{M}$_2$, \texttt{L} is not in \texttt{X}.
\end{theorem}

\begin{proof}[Proof of Theorem \ref{model_subset}]
	Let \texttt{P} be a program, and \texttt{M} a partial model of \texttt{P} generated by the s(ASP) algorithm. Let \texttt{M}$_2$ and \texttt{P}$_2$ be the results of grounding \texttt{M} and \texttt{P}, respectively, over the s(ASP) universe. If \texttt{M} contains loop variables then we may choose a domain for each loop variable that does not contradict the rest of \texttt{M} without loss of generality. This is just selecting one out of the infinite number of partial models represented by \texttt{M}. Assume that there exists at least one assignment that contains no empty variables since we consider such a situation as a failure.
	
	Before proving our claim we must show that \texttt{M}$_2$ is consistent. That is, for some literal \texttt{L} it is not the case that \texttt{L} and \texttt{not L} are both in \texttt{M}$_2$. First, notice that if the value of \texttt{L} depends on \texttt{not L} (and visa versa) then the s(ASP) algorithm will fail or the goal will be constrained so that \texttt{L} and \texttt{not L} will not be in the grounding. This is because it is an odd cycle over negation. So we only need to consider the case where \texttt{G} $\in$ goal(\texttt{L}) unifies with something in \texttt{M}, but we want to prove a goal that unifies with the negation of \texttt{G}. However, by lemma \ref{ccf} we know that the second goal will be restricted so that it no longer unifies with \texttt{L} or \texttt{not L}. So, \texttt{M}$_2$ is consistent.
	
	Let \texttt{P}$_3$ be the result of trimming \texttt{P}$_2$ with respect to \texttt{M}$_2$, and  \texttt{S} be a set of literals such that \texttt{L} $\in$ \texttt{S} $\iff$ \texttt{L} $\in$ \texttt{M}$_2$ $\vee$ \texttt{not L} $\in$ \texttt{M}$_2$. \texttt{S} is a splitting set of \texttt{P}$_3$. Now we must do two things. First we must show that \texttt{M}$_2$ is a stable model of the bottom, and that there exists a stable model for the top.
	
	To prove that \texttt{M}$_2$ is a stable model of the bottom we must show that:
	\begin{enumerate}
		\item All literals in \texttt{M}$_2$ are in the least model of the reduct for \texttt{P}$_3$, and 
		\item No literal \texttt{L} with \texttt{not L} in \texttt{M}$_2$ will be in it. 
	\end{enumerate}
	
	\begin{description}
		\item[Case 1:] For all literals \texttt{L} $\in$ \texttt{M}$_2$: Let \texttt{L}$^\prime$ be the non-ground atom in \texttt{M} that is used to generate \texttt{L}, and \texttt{R} be the clause in \texttt{P} that is used to prove \texttt{L}$^\prime$. We can construct a tree by using \texttt{L}$^\prime$ as the root, and the body literals from \texttt{R} as the children. The negated goals will be in \texttt{M} and later removed from the clause when computing the reduct. So, we can ignore them and only consider literals as children. Additionally, we will keep the groundings and constraints of the variables at the time of success. Then ground the tree such that the resulting tree has \texttt{L} as the root. This corresponds to a clause in \texttt{P}$_3$, since it would have been marked and therefore not removed. The leaves of such a tree must have facts in the reduct of \texttt{P}$_3$, and therefore will be in the least model. From there we know that the root of each level going up the tree will be in the least model, including \texttt{L}.
		\item[Case 2:] For all literals \texttt{L} such that \texttt{not L} $\in$ \texttt{M}$_2$ we must show that there is no way \texttt{L} can be in the least model of the reduct for \texttt{P}$_3$. Firstly, if a clause with \texttt{L} as the head is part of a positive cycle for \texttt{L}, then it cannot be used to put \texttt{L} into the least model. So, we only need to consider noncyclic cases. For these cases we will prove it inductively, and to do that we will define the \textbf{level} of a clause. If a clause is a fact then it has level zero. For non-fact clauses, we say that a body goal \texttt{B} has a level equal to that of the highest level clause with atom(\texttt{B}) as the head. The level of all non-facts is one plus the highest level of the body literals. In the case of a body literal for which there are no clauses, it is considered level zero.
		\begin{description}
			\item[Base Case:] Let \texttt{L} be a literal such that \texttt{not L} $\in$ \texttt{M}$_2$. There cannot be a fact for \texttt{L}, otherwise \texttt{not L} could not be in \texttt{M}$_2$. The goal \texttt{not L} comes from the success of a dual rule, which would always fail if a fact for L existed. If \texttt{L} has no clauses, then it cannot be in the least model.
			\item[Inductive Hypothesis:] Let \texttt{L} be a literal such that \texttt{not L} $\in$ \texttt{M}$_2$. Suppose all clauses with a level less than or equal to \texttt{k} cannot be used to place \texttt{L} in the least model.
			\item[Inductive Step:] Let \texttt{L} be a literal such that \texttt{not L} $\in$ \texttt{M}$_2$. Let \texttt{R} be a clause with \texttt{L} in the head and a level of \texttt{k} + 1. In order for the dual to succeed and allow \texttt{not L} to be in the grounding there must be a goal \texttt{G} in \texttt{R} such that \texttt{G} $\not\in$ \texttt{M}$_2$. Since \texttt{G} $\not\in$ \texttt{M}$_2$ but was not trimmed from \texttt{R} the negation of \texttt{G} must be in \texttt{M}$_2$. If \texttt{G} is negated then \texttt{R} would be removed when computing the reduct. So, we only need to consider the case \texttt{G} is not negated. \texttt{G} must have a level of at most \texttt{k}, and by the inductive hypothesis we know that there is no way to place \texttt{G} into the least model of the reduct, and therefore \texttt{R} cannot be used to place \texttt{L} into the least model.
		\end{description}
		Thus, by induction \texttt{L} is not in the least model of the reduct for \texttt{P}$_3$.
	\end{description}
	Therefore, \texttt{M}$_2$ is a stable model of the bottom of \texttt{P}$_3$ with respect to the splitting set \texttt{S}.
	
	Now we must show that there exists a stable model for the top. We will do this by observing that the only way for there not to be a stable model is if there is an inconsistency, and there can be an inconsistency only if there is an odd cycle. So, we will show that the modified program from the top will contain no odd cycles. 
	
	Let \texttt{R} be an OLON in \texttt{P}$_2$, and \texttt{R}$^\prime$ the clause in \texttt{P} such that when grounding \texttt{R}$^\prime$ over the s(ASP) universe, \texttt{R} is generated. Since \texttt{R} is part of an odd cycle, \texttt{R}$^\prime$ is also considered part of an odd cycle since we only look at predicate name and arity. Thus there will be a NMR check for \texttt{R}$^\prime$. By lemma \ref{forall}, we can treat the NMR check as a conjunction of checks with the head grounded over the s(ASP) universe. So, either there exists a body literal in \texttt{R} with its negation in \texttt{M}$_2$ or the head of \texttt{R} is in \texttt{M}$_2$. In the second case, \texttt{R} will either be removed through trimming or will be in the bottom of \texttt{P}$_3$. For the first case, assume \texttt{R} is not removed through trimming or in the bottom of \texttt{P}$_3$. Then \texttt{R} will be removed when computing the partial evaluation for the top since the negation of some literal in the body is in \texttt{M}$_2$. Thus, there are no odd cycles when computing the answer sets of the top.
	
	It is apparent from the splitting theorem that if \texttt{L} $\in$ \texttt{M}$_2$ is a literal then \texttt{L} is in \texttt{X}. So, we only need to show that if \texttt{not L} $\in$ \texttt{M}$_2$ then \texttt{L} $\not\in$ \texttt{X}. First, notice that the truth value of \texttt{L} is determined by the bottom of \texttt{P}$_3$ with respect to \texttt{S}, and cannot be in the stable model of the top. Thus, \texttt{L} cannot be in \texttt{X}.
	
	By lemma \ref{unneeded_clauses}, we know that a stable model for \texttt{P}$_3$ is also a stable model of \texttt{P}$_2$.
\end{proof}

\subsubsection{Completeness}

While s(ASP) is sound for the set of all legal programs, completeness for this 
set is impossible. While we have striven to make s(ASP) complete for the 
largest class of programs possible, we leave the precise definition of this 
class and the associated proofs to future work.\footnote{The exception being 
finite, legal, grounded programs, for which we have proven completeness in 
Section \ref{sec:groundcorrect}.} Instead, we argue that the utility of our 
method outweighs its lack of completeness.

It is easy enough to show that s(ASP) cannot possibly be complete for all legal 
programs. Consider the class of \textit{stratified programs}, that is, programs 
with no loops over negation. A stratified program will always have a unique 
stable model which coincides with its \textit{perfect model}, the model 
produced by the perfect model semantics \cite{cadoli93}. However, the perfect 
model of such a program may be incomputable \cite{apt1990}. Therefore, even 
though a stable model must exist for such a program, we may be unable to 
compute it. As such, it is not possible to guarantee completeness for such 
programs.

The trade-off for this loss of completeness is a massive increase in expressive 
power. The propositional stable model semantics can only express relations 
which are co-NP, however, the predicate stable model semantics and s(ASP) can 
express relations which are $\Pi^1_1$ \cite{schlipf1990,cadoli93}.

In addition to increased computational expressiveness, s(ASP) supports lists, 
complex data structures and real numbers, providing programmers with tools not 
found in any other implementation of the stable model semantics. With these 
features, even programs which can already be expressed in the propositional 
semantics may be easier to write in s(ASP).

Aside from completeness itself, the only ``desirable'' property that we lose 
compared to other implementations of the stable model semantics is the 
guarantee that a program will always terminate. However, guaranteed termination 
is a double-edged sword. It implies that only decidable problems can be 
encoded, as, by definition, this guarantee cannot be applied to semi-decidable 
or undecidable problems. By abandoning guaranteed termination, we are able to 
support programs which encode such problems, something that no other 
implementation of the stable model semantics can claim. This is significant, as 
problems which are undecidable in general may still produce useful results for 
some cases.

Thus, we trade completeness for superior functionality and the ability to 
encode problems which no other implementation of the stable model semantics can 
handle. While completeness is certainly desirable, it is our firm belief that 
the gains derived from this trade significantly outweigh the losses.

\section{Implementation and Examples} \label{sec:imp}

A fully functional prototype implementation of the method presented here has 
been created, also using the name s(ASP). The implementation is written in 
Prolog and totals about 4,300 lines of code (excluding comments and blank 
lines). An open source release is available at \cite{saspweb}.

Unlike its predecessor, Galliwasp, s(ASP) is completely self-contained: neither 
a grounder nor a separate compiler is required. As with our method, the 
prototype will accept any legal normal logic program and execute it 
\textit{without grounding any portion of the program at any stage}. While the 
prototype is not designed to be competitive in terms of speed, our method 
allows it to offer features not found in any other implementation of the stable 
model semantics, including answer set programming systems. In the following 
subsections, we will look at how s(ASP) behaves with a number of examples.

\subsection{Example: N Queens with Lists} \label{sec:example}

A variant of the N queens problem using lists, can be found in Figure 
\ref{fig:queens}. This example is of particular interest, as it has no finite 
grounding and thus cannot be run by other implementations of the stable model 
semantics. Additionally, the even loop in the last two lines of the code will 
produce two loop variables, discussed in Section \ref{sec:even}.

\begin{figure}
	\figrule
	\begin{center}
		\begin{verbatim}
		% solve the N queens problem for a given N, returning a list of queens as Q
		nqueens(N, Q) :-
		    nqueens(N, N, [], Q).
		
		% pick queens one at a time and test against all previous queens
		nqueens(X, N, Qi, Qo) :-
		    X > 0,
		    pickqueen(X, Y, N),
		    not attack(X, Y, Qi),
		    X1 is X - 1,
		    nqueens(X1, N, [q(X, Y) | Qi], Qo).
		nqueens(0, _, Q, Q).
		
		% pick a queen for row X.
		pickqueen(X, Y, Y) :-
		    Y > 0,
		    q(X, Y).
		pickqueen(X, Y, N) :-
		    N > 1,
		    N1 is N - 1,
		    pickqueen(X, Y, N1).
		
		% check if a queen can attack any previously selected queen
		attack(X, _, [q(X, _) | _]). % same row
		attack(_, Y, [q(_, Y) | _]). % same col
		attack(X, Y, [q(X2, Y2) | _]) :- % same diagonal
		    Xd is X2 - X, abs(Xd, Xd2),
		    Yd is Y2 - Y, abs(Yd, Yd2),
		    Xd2 = Yd2.
		attack(X, Y, [_ | T]) :-
		    attack(X, Y, T).
		
		q(X, Y) :- not negq(X, Y).
		negq(X, Y) :- not q(X, Y).
		
		abs(X, X) :- X >= 0.
		abs(X, Y) :- X < 0, Y is X * -1.
		\end{verbatim}
	\end{center}
	\caption{N Queens Program with Lists}
	\label{fig:queens}
	\figrule
\end{figure}

When executed by our prototype implementation the user will get the following:
\begin{verbatim}
?- nqueens(5,X).
{ nqueens(5,[q(1,2),q(2,4),q(3,1),q(4,3),q(5,5)]), q(1,2), q(2,4),
q(3,1), q(4,3), q(5,5) }
X = [q(1,2),q(2,4),q(3,1),q(4,3),q(5,5)].

?- nqueens(4,X).
{ nqueens(4,[q(1,2),q(2,4),q(3,1),q(4,3)]), q(1,2), q(2,4), q(3,1),
q(4,3) }
X = [q(1,2),q(2,4),q(3,1),q(4,3)];
{ nqueens(4,[q(1,3),q(2,1),q(3,4),q(4,2)]), q(1,3), q(2,1), q(3,4),
q(4,2) }
X = [q(1,3),q(2,1),q(3,4),q(4,2)];
false.
\end{verbatim}

As no grounding is performed, multiple instances of the problem can be queried 
in a single session. The sample output illustrates this by querying both four 
and five queens. While the entire partial stable model is provided, variables 
in the query are printed, making desired information much easier to find. In 
the above example, \texttt{X} will be bound to the list of queens selected. As 
with Prolog interpreters, `;' and `.' can be used to reject or accept a 
solution, respectively. As there are only two solutions for four queens, 
pressing `;' a second time leads to failure. Note that the output will often 
contain variables with names consisting of ``Var'' followed by an integer. This 
is simply because s(ASP) renames variables to ensure that they are unique.

\subsection{Example: Hamiltonian Cycle Detection}

\begin{figure}
	\figrule
	\begin{center}
		\begin{verbatim}
		reachable(V) :- chosen(U, V), reachable(U).
		reachable(0) :- chosen(V, 0).
		
		% Every vertex must be reachable.
		:- vertex(U), not reachable(U).
		
		% Choose exactly one edge from each vertex.
		other(U, V) :-
		    vertex(U), vertex(V), vertex(W),
		    V \= W, chosen(U, W).
		chosen(U, V) :-
		    vertex(U), vertex(V),
		    edge(U, V), not other(U, V).
		
		% Two edges cannot be incident on the same
		% vertex.
		:- chosen(U, W), chosen(V, W), U \= V.
		
		% Sample graph: vertexes and the edges connecting them.
		vertex(0).
		vertex(1).
		vertex(2).
		vertex(3).
		vertex(4).
		
		edge(0, 1).
		edge(1, 2).
		edge(2, 3).
		edge(3, 4).
		edge(4, 0).
		edge(4, 1).
		edge(4, 2).
		edge(4, 3).
		\end{verbatim}
	\end{center}
	\caption{A program for Hamiltonian cycle detection with a simple graph 
		included.}
	\label{fig:hamiltonian}
	\figrule
\end{figure}

Figure \ref{fig:hamiltonian} contains an encoding of the Hamiltonian cycle 
problem, along with a simple graph. The results of this example provide an 
interesting look at our use of negatively constrained variables in output. 
While another solution exists, due to space limitations, only the first is 
provided. The cycle is represented by the \texttt{chosen/2} elements at the 
beginning of the set:
\begin{verbatim}
?- reachable(0).
{ chosen(0,1), chosen(1,2), chosen(2,3), chosen(3,4), chosen(4,0), 
edge(0,1), edge(1,2), edge(2,3), edge(3,4), edge(4,0), edge(4,1), 
edge(4,2), edge(4,3), other(0,0), other(0,2), other(0,3), 
other(0,4), other(1,0), other(1,1), other(1,3), other(1,4), 
other(2,0), other(2,1), other(2,2), other(2,4), other(3,0), 
other(3,1), other(3,2), other(3,3), other(4,1), other(4,2), 
other(4,3), other(4,4), reachable(0), reachable(1), reachable(2), 
reachable(3), reachable(4), vertex(0), vertex(1), vertex(2), 
vertex(3), vertex(4), not chosen(0,0), not chosen(0,2), not 
chosen(0,3), not chosen(0,4), not chosen(0,Var644) ( Var644 \= 0, 
Var644 \= 1, Var644 \= 2, Var644 \= 3, Var644 \= 4 ), not 
chosen(1,0), not chosen(1,1), not chosen(1,3), not chosen(1,4), not 
chosen(1,Var710) ( Var710 \= 0, Var710 \= 1, Var710 \= 2, Var710 \= 
3, Var710 \= 4 ), not chosen(2,0), not chosen(2,1), not chosen(2,2), 
not chosen(2,4), not chosen(2,Var776) ( Var776 \= 0, Var776 \= 1, 
Var776\= 2, Var776 \= 3, Var776 \= 4 ), not chosen(3,0), not 
chosen(3,1), not chosen(3,2), not chosen(3,3), not chosen(3,Var842) 
( Var842 \= 0, Var842 \= 1, Var842 \= 2, Var842 \= 3, Var842 \= 4 ), 
not chosen(4,1), not chosen(4,2), not chosen(4,3), not chosen(4,4), 
not chosen(4,Var908) ( Var908 \= 0, Var908 \= 1, Var908 \= 2, Var908 
\= 3, Var908 \= 4 ), not chosen(Var627,_) ( Var627 \= 0, Var627 \= 
1, Var627 \= 2, Var627 \= 3, Var627 \= 4 ), not chosen(Var663,1) ( 
Var663 \= 0, Var663 \= 1, Var663 \= 2, Var663 \= 3, Var663 \= 4 ), 
not chosen(Var734,2) ( Var734 \= 0, Var734 \= 1, Var734 \= 2, Var734 
\= 3, Var734 \= 4 ), not chosen(Var805,3) ( Var805 \= 0, Var805 \= 
1, Var805 \= 2, Var805 \= 3, Var805 \= 4 ), not chosen(Var876,4) ( 
Var876 \= 0, Var876 \= 1, Var876 \= 2, Var876 \= 3, Var876 \= 4 ), 
not chosen(Var922,0) ( Var922 \= 0, Var922 \= 1, Var922 \= 2, Var922 
\= 3, Var922 \= 4 ), not edge(0,0), not edge(0,2), not edge(0,3), 
not edge(0,4), not edge(1,0), not edge(1,1), not edge(1,3), not 
edge(1,4), not edge(2,0), not edge(2,1), not edge(2,2), not 
edge(2,4), not edge(3,0), not edge(3,1), not edge(3,2), not 
edge(3,3), not edge(4,4), not other(0,1), not other(1,2), not 
other(2,3), not other(3,4), not other(4,0), not vertex(Var31) ( 
Var31 \= 0, Var31 \= 1, Var31 \= 2, Var31 \= 3, Var31 \= 4 ) } . 
\end{verbatim}

\section{Applications} \label{sec:applications}

The s(ASP) system is publicly available \cite{saspweb}, and has been used 
to develop a number of non-trivial applications based on ASP; it has also been 
used to organize an AI hackathon \cite{hackathon}. Some of these applications 
cannot be executed on traditional ASP systems such as CLASP, as these 
applications make use of lists and structures to represent information. They 
have been developed by people who are not experts in ASP. These applications 
include:

\begin{enumerate}
\item {\bf Degree Audit System:} A system for automatically performing a degree 
audit of a student's undergraduate transcript at a US University, i.e., 
automatically determining whether a student can graduate with a degree or not, 
has been developed using the s(ASP) system \cite{degreeaudit}. The system 
represents the graduation requirements laid out in the course catalog as ASP 
clauses. Use of negation is important for representing these requirements. The 
system has to make use of lists, and has hundreds of courses that appear as 
constants in the program (hence its grounding will produce an inordinately 
large program).

\item {\bf Physician Advisory System:} A system for disease management, 
particularly, for chronic heart failure has been developed using the s(ASP) 
system \cite{chf}. This system automates the 80-page guidelines (that the 
American College of Cardiology has developed) by representing them in ASP. 
While the current system can be run under systems such as CLASP due to the
number of constants not being too large, the final system that models a 
doctor's full knowledge will have quite a few constants, and advanced 
data-structures may be needed.

\item {\bf Automating Textbook Knowledge:} A system that represents 
high-school level knowledge about cells (in the discipline of biology) as
answer set programs has been developed using s(ASP). It can answer high-school 
level questions posed as s(ASP) queries. The goal is to represent the knowledge 
in the entire introductory biology textbook as an answer set program, and then 
be able to automatically answer questions that would be asked of a student (the 
questions have to be translated into ASP queries that are then executed to find 
the answer).

\item {\bf Birthday Gift Advisor:} A recommendation system for birthday gifts 
has also been developed using the s(ASP) system. This system codes a person's 
knowledge about friends, level of friendship, a person's wealth level, 
generosity level, and hobbies as answer set programs. When queried, the system 
can recommend a birthday present for a particular friend (e.g., on one's 
Facebook page). Note that other similar recommendation systems can 
also be built using s(ASP).
\end{enumerate}

\section{Related and Future Work} \label{sec:related}

With respect to related work, most of it focuses on answer set programming 
rather than purely on stable models. Perhaps most notably, DLV supports lists and structures, however, the underlying execution mechanism is still based on grounding and then finding the stable models of the resulting propositional program \cite{dlv-manual,dlv-impl,dlv-sys}. To ensure that the grounded program stays finite, the DLV system resorts to techniques such as \textit{finite domain checking} and requiring that the programmer specify an \textit{upper integer limit} \cite{dlv-manual}, i.e., the maximum numerical value allowed in the program. In contrast, no finite domain checks or upper integer limits are required by the s(ASP) system.

Various other attempts have been made to achieve 
a goal-directed method for executing propositional answer set programs 
\cite{alferes,bonatti,bonatti2,pp05,lmp09,syy04}. However, all of these either 
alter the semantics or significantly restrict the programs accepted, and all 
are restricted to grounded programs.
Our own method will accept any normal logic program without altering the underlying stable model semantics. Similarly, work has 
been done in the area of predicate answer set programming, but only with much 
more severe restrictions on accepted programs \cite{Bonatti2004,Heymans03}. 
Other ASP systems ground ``on the fly'', but grounding is still performed at 
some point during execution \cite{dalpalu09,daotran12,lefevre09a,lefevre09b}.
A variety of efforts have focused on constructive negation 
\cite{chan,stuckey,pearce}, but our method's combination of negatively 
constrained variables and specially adapted dual rules represents a unique 
approach.

There are a vast number of potential routes for future work, including 
practical applications and extensions to the language. In 
particular, the restriction that left recursion cannot lead to success should 
be resolvable through the addition of tabling \cite{swift}. We 
also plan to invest effort on improving our implementation's efficiency. This 
will be achieved by designing a WAM-style abstract machine specialized for our 
method and developing an emulator for it. Normal logic programs can then be 
compiled to this abstract machine and executed using the emulator, much in the 
style of modern Prolog systems.

\section{Conclusion} \label{sec:conclusion}

In this paper, we have presented a method for computing partial stable models 
of predicate normal logic programs and proven it sound for a large class of 
programs. The key to this method is the use of a special, non-Herbrand universe 
which allows us to ensure soundness while still producing useful results for a 
large class of programs. Our method also relies on coinduction and constructive 
negation to execute programs in a top-down manner similar to that used in 
Prolog systems. Compared to similar attempts, our method supports a much larger 
class of programs. Indeed, only three restrictions are placed on input 
programs: that arithmetic operations must be ground when executed, that left 
recursion cannot lead to success, and that two negatively constrained variables 
cannot be disunified with each other. An implementation of our method, s(ASP), 
is freely available \cite{saspweb} and has already been used in the development 
of non-trivial applications.

\section*{Acknowledgment}

We are grateful to Howard Blair for his help and input, particularly on the 
subject of completeness.

\bibliographystyle{acmtrans}
\bibliography{sasp}

\end{document}